\documentclass[journal, a4paper]{IEEEtran}
\IEEEoverridecommandlockouts

\usepackage{cite}
\usepackage{amsmath,amssymb,amsfonts, amsthm}
\usepackage{algorithmic}
\usepackage{graphicx}
\usepackage{textcomp}
\usepackage{xcolor}
\usepackage{booktabs}
\usepackage{tabularx}
\usepackage{algorithm}
\usepackage{algorithmic}
\usepackage{bbm}
\usepackage{bm}
\usepackage{multirow}
\usepackage{float}
\usepackage{glossaries}
\usepackage{siunitx}
\usepackage{accents}
\newacronym{3gpp}{3GPP}{3rd Generation Partnership Project}    
\newacronym{6g}{6G}{sixth generation of mobile networks}
\newacronym{awgn}{AWGN}{additive white Gaussian noise}
\newacronym{ber}{BER}{bit error rate}
\newacronym{bler}{BLER}{block error rate}
\newacronym{cdf}{CDF}{cumulative distribution function}
\newacronym{cml}{CML}{commercial microwave link}
\newacronym{iot}{IoT}{Internet of Things}
\newacronym{isac}{ISAC}{integrated sensing and communications}
\newacronym{ntn}{NTN}{Non-Terrestrial Network}
\newacronym{leo}{LEO}{Low Earth Orbit}
\newacronym{geo}{GEO}{Geosynchronous Earth Orbit}
\newacronym{isl}{ISL}{Inter-Satellite Link}
\newacronym{gsl}{GSL}{Ground-to-Satellite Link}
\newacronym{qos}{QoS}{Quality of Service}
\newacronym{ofdma}{OFDMA}{Orthogonal Frequency-Division Multiple Access}
\newacronym{b5g}{B5G}{5G and Beyond}
\newacronym{mimo}{MIMO}{Multiple-Input Multiple-Output}
\newacronym{mcs}{MCS}{modulation and coding scheme}
\newacronym{embb}{eMBB}{Enhanced Mobile Broadband}
\newacronym{ue}{UE}{User Equipment}
\newacronym{nr}{NR}{New Radio}
\newacronym{gap}{GAP}{Generalized Assignment Problem}
\newacronym{mgap}{MGAP}{Multi-Level Generalized Assignment Problem}
\newacronym{ml}{ML}{machine learning}
\newacronym{csi}{CSI}{channel state information}    
\newacronym{ran}{RAN}{radio access network}
\newacronym{5g}{5G}{the 5th generation of mobile networks}
\newacronym{uav}{UAV}{unmanned aerial vehicle}
\newacronym{snr}{SNR}{signal-to-noise ratio}
\newacronym{ra}{RA}{resource allocation}
\newacronym{rssi}{RSSI}{received signal strength indicator}
\newacronym{rv}{RV}{random variable}
\newacronym{ul}{UL}{uplink}
\newacronym{dl}{DL}{downlink}
\newacronym{mle}{MLE}{maximum likelihood estimator}
\newacronym{crb}{CRB}{Cramér-Rao bound}
\newacronym{mse}{MSE}{mean-squared error}
\newacronym{nmse}{NMSE}{normalized mean-squared error}
\newacronym{ppp}{PPP}{Poisson point process}
\newacronym{milp}{MILP}{mixed-integer linear problem}
\newacronym{rtt}{RTT}{round-trip time}
\newacronym{noma}{NOMA}{non-orthogonal multiple access}
\newacronym{jmra}{JMRA}{joint matching and resource allocation}
\newacronym{dmrab}{DMRAB}{disjoint matching and resource allocation benchmark}
\newacronym{kpi}{KPI}{key performance indicator}
\newacronym{dtmc}{DTMC}{discrete-time Markov chain}
\newacronym{sca}{SCA}{successive convex approximation}
\usepackage[font=footnotesize]{caption}
\usepackage[font=scriptsize]{subcaption}
\usepackage{tikz}
\usepackage{pgfplots}
\pgfplotsset{compat=1.18}
\usepgfplotslibrary{colorbrewer}
\usetikzlibrary{automata,arrows,positioning,calc,shapes.misc}

\theoremstyle{definition}
\newtheorem{definition}{Definition}[]

\newtheorem{proposition}{Proposition}
\newcommand{\set}{\mathcal}

\DeclareMathOperator*{\round}{round}
\DeclareMathOperator*{\argmax}{arg\,max}

\def\BibTeX{{\rm B\kern-.05em{\sc i\kern-.025em b}\kern-.08em
    T\kern-.1667em\lower.7ex\hbox{E}\kern-.125emX}}

\newcommand{\change}[1]{#1}

\definecolor{amaranth}{rgb}{0.9, 0.17, 0.31}

\newcommand{\mc}[1]{\mathcal{#1}}   

\newlength{\dhatheight}
\newcommand{\hathat}[1]{%
    \settoheight{\dhatheight}{\ensuremath{\hat{#1}}}%
    \addtolength{\dhatheight}{-0.35ex}%
    \hat{\vphantom{\rule{1pt}{\dhatheight}}%
    \smash{\hat{#1}}}}

\begin{document}
\bstctlcite{IEEEexample:BSTcontrol}
\title{\change{Integrating Atmospheric} Sensing and Communications for Resource Allocation in \change{NTNs}}


\author{Israel Leyva-Mayorga,~\IEEEmembership{Member,~IEEE,} 
        Fabio Saggese,~\IEEEmembership{Member,~IEEE,} 
        Lintao Li,  and 
        Petar Popovski~\IEEEmembership{Fellow,~IEEE}
\thanks{I. Leyva-Mayorga, F. Saggese, and P. Popovski are with the Department of Electronic Systems, Aalborg University, Denmark; email: \{ilm, fasa, petarp\}@es.aau.dk.  L. Li is with the Department of Electronic Engineering, Tsinghua University, Beijing 100084, China.
This work was partly supported by the Villum Investigator grant ``WATER'' from the Villum Foundation, Denmark.}}

\maketitle
\begin{abstract}
The integration of \glspl{ntn} with \gls{leo} satellite constellations into 5G and Beyond is essential to achieve truly global connectivity. A distinctive characteristic of \gls{leo} mega-constellations is that they constitute a global infrastructure with predictable dynamics, which enables the pre-planned allocation of radio resources. However, the different bands that can be used for ground-to-satellite communication are affected differently by atmospheric conditions such as precipitation, which introduces uncertainty on the attenuation of the communication links at high frequencies. Based on this, we present a compelling case for applying \gls{isac} in heterogeneous and multi-layer \gls{leo} satellite constellations over wide areas. Specifically, \change{we propose a sensing-assisted communications framework and frame structure that not only enables the accurate estimation of the \emph{atmospheric} attenuation in the communication links through sensing but also leverages this information to determine the optimal serving satellites and allocate resources efficiently for downlink communication with users on the ground.}
The results show that, by dedicating an adequate amount of resources for sensing and solving the association and resource allocation problems jointly, it is feasible to increase the average throughput by $59$\% and the fairness by $700$\% when compared to solving these problems separately.
\end{abstract}
\glsresetall

\glsresetall

\section{Introduction}
\label{sec:intro}
\Glspl{ntn} are key for the \gls{6g} to overcome the limitations of terrestrial infrastructure by enabling ubiquitous connectivity, ensuring global access, and enhancing network resilience by offering alternative communication pathways in case of terrestrial network failure~\cite{IMT2030}. 
Especially, \gls{leo} satellites, positioned at altitudes ranging from approximately $500$ to $2000$ kilometers, have gained unprecedented importance for \gls{6g} \glspl{ntn}, as they can achieve significantly lower latency compared to \gls{geo} satellites, which is crucial for time-sensitive applications~\cite{Leyva-Mayorga2020}.
Furthermore, the deployment of \gls{leo} satellite constellations, comprising numerous interconnected satellites, ensures redundancy and reliability. If one satellite fails, others can seamlessly take over, enhancing the \gls{ntn} resilience~\cite{Rinaldi2020:ntn-survey}. 

\begin{figure}[thb]
    \centering
    \includegraphics[width=.9\columnwidth]{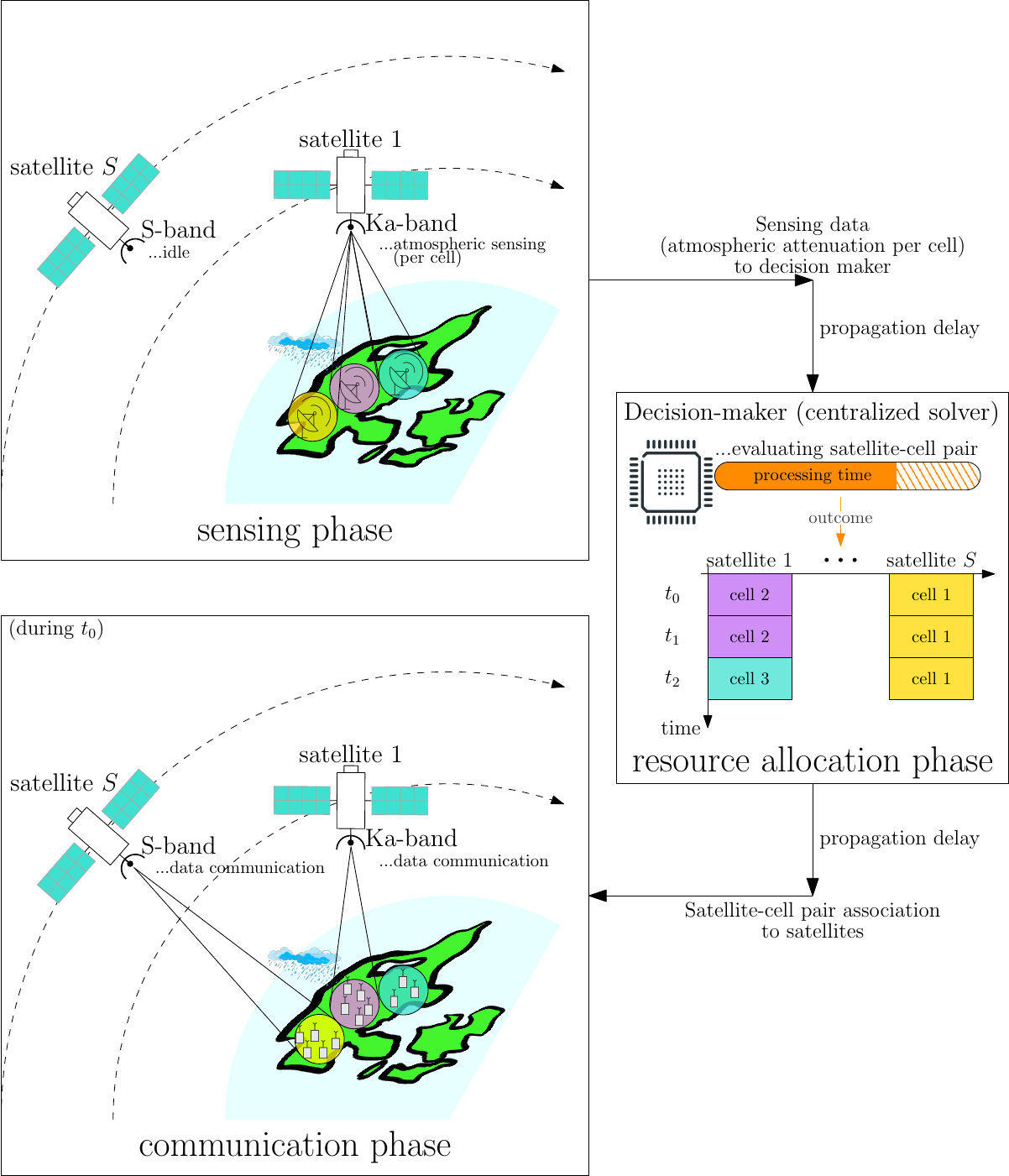}
    \caption{\change{Toy example of the scenario at hand and the algorithmic depiction of the integration of sensing and communication.}}
    \label{fig:toy}
\end{figure}

Inspired by the terrestrial network, the \gls{3gpp} proposes to divide the ground segment of a 5G \gls{ntn}, where \glspl{ue} are located, into cells served by different \gls{leo} satellites of the constellation~\cite{3GPPTR38.821}.
Thus, the \gls{qos} experienced by the ground \gls{ue} depends on the satellite-cell association, the \gls{snr} at the receiver~\cite{Rinaldi2020:ntn-survey}, and the resources assigned for communication. Consequently, the \gls{qos} of the \glspl{ue} can be maximized by performing an optimal satellite-cell association and \gls{ra}; however, this requires an accurate model of the \glspl{gsl} and their state. 

By knowing the geographical position of the satellite-cell pairs, \change{a free-space propagation model with} \gls{awgn} can be employed to model the \glspl{gsl}. Nevertheless, this model neglects all the \emph{atmospheric effects} as sources of uncertainty. Specifically, K- and higher bands have a greater available bandwidth than, for example, the S-band, but these are greatly attenuated by the presence of water particles. Therefore, the attenuation due to atmospheric effects, further referred to as \emph{atmospheric attenuation}, can easily limit the communication performance in higher bands if not accurately estimated~\cite{ITU-rain}.
This builds the case for enabling \gls{isac} for \gls{ntn} by integrating a sensing functionality into the \gls{ntn} system to estimate the impact of atmospheric effects on the \gls{snr}. This goes in line with the vision of \gls{isac} playing an essential role in \gls{6g} by combining the functionalities of wireless communication and environmental sensing into a single system~\cite{Liu2022:isac-survey}.

Specifically, we focus on \emph{sensing-assisted communication} to leverage advanced sensing techniques for estimating the \gls{snr} and adapting the \gls{ra} to variations caused by atmospheric effects. 
From a communication point-of-view, the \gls{ra} process uses the sensing information to optimize the association of satellites to cells, accounting for the impact of atmospheric conditions on different frequency bands within the satellite constellation.
Additionally, from a \emph{sensing-as-a-service} perspective~\cite{Liu2022:isac-survey}, the estimated attenuation can provide real-time weather information at a global continental scale, supporting traditional systems in rainfall monitoring, considered one of the potential use-cases defined by the \gls{3gpp} for \gls{isac}~\cite{3GPPTR22.837isac}. 

When considering \gls{leo} \gls{ntn}, further investigation is needed on \gls{isac} protocols orchestrating sensing and communication coexistence. First, only satellites operating on K- or above bands can perform atmospheric sensing, as only these are attenuated by the presence of water molecules~\cite{gritton2005atmospheric}. Second, the communication performance depends on the \gls{ra} process, which, in turn, depends on the sensing process. Third, the \gls{ra} must allow the satellites at K- and higher bands to exploit their high available bandwidth without suffering from attenuation, while the rest of the satellites use their narrower bandwidth to cover the remaining cells. 

To address these challenges, we propose to organize the \gls{ntn} operations into sensing, \gls{ra}, and communication phases, as depicted in Fig.~\ref{fig:toy}: within the sensing phase, sensing-enabled satellites 
\change{estimate the atmospheric attenuation in each cell} (e.g., \change{by} transmitting pilot sequences to anchor nodes), while other satellites remain idle. Then, the sensing data is relayed to the decision-makers, who use the \change{estimated atmospheric attenuation at each cell to perform} the \gls{ra} and communicate \change{the resulting satellite-pair association} to the satellites. \change{In our Fig.~\ref{fig:toy} example, cell $1$ is severely affected by rain; by $a$) sensing and estimating a severe attenuation for the Ka-band satellite $1$, and $b$) knowing the free-space path loss for satellites $1$ and $S$, the decision-maker calculates that the optimal solution is to serve cell $1$ with satellite $S$ and the rest of the cells with satellite $1$.} Finally, in the communication phase, all the satellites serve the \glspl{ue} in their allocated cells.
Naturally, the allocation of the sensing and \gls{ra} phases must be tailored to the dynamics of the environment, as sensing and \gls{ra} phases must be scheduled to account for the changes in the atmospheric effects and the \glspl{gsl} conditions over time. Furthermore, as Fig.~\ref{fig:toy} shows, a real-time implementation of the \gls{ra} must consider the propagation delays for relaying the sensing data to the decision-makers and the result to the satellites, as well as the delay for computing the \gls{ra}.
To the best of our knowledge, this is the first work that comprehensively addresses the coexistence between atmospheric sensing and communication in a \gls{leo} \gls{ntn} with heterogeneous satellites. The main contributions of this work are: 
\begin{enumerate}
    \item We build a case for \gls{isac} in \gls{leo} \glspl{ntn} by presenting a framework \change{for enabling and exploiting atmospheric sensing} for continent-wide \change{sensing-assisted} communications with multiple satellites. The framework considers both the time needed to transmit the sensing data and the time to find the optimal resource allocation based on the sensed data. To the best of our knowledge, this is the first study for \change{sensing-assisted} communications that considers heterogeneous satellites with different altitudes and operating at vastly diverse frequency bands.
    
    \item In this setting, we \change{tackle the joint satellite-to-cell association (i.e., matching) and \gls{ra} tasks by formulating and solving an optimization problem that aims to guarantee a proportionally fair solution accounting for the service interruption time due to 5G \gls{ntn} handovers. To solve this joint matching and \gls{ra} problem, we combine continuous relaxation, successive convex approximation, and the iterative method of multipliers to achieve an increase in Jain's fairness index up to $700\%$ when compared to a disjoint approach where the matching and \gls{ra} problems are solved independently.}
    
    \item We present a thorough performance analysis of our framework under practical timing constraints. In particular, we show that the immediate use of the sensed data for resource allocation is infeasible due to propagation and processing delays and evaluate the performance loss due to realistic processing times when compared to an idealistic case where these delays are not considered.
    
    \item We illustrate the trade-off between sensing accuracy and communication performance in a scenario with a realistic spatial-temporal rain model. \change{Specifically, it is observed that the average per-user throughput with a practical pilot length of $2^{12}$ is only $1\%$ lower than the idealistic case with perfect \gls{csi} and over $14\%$ higher than without atmospheric sensing.}
\end{enumerate}
The present paper greatly extends our preliminary work on \gls{ra} in \gls{leo} \gls{ntn}~\cite{Mayorga2023icc} by considering heterogeneous satellites, including a sensing mechanism, designing an \gls{isac} frame, and accurately accounting for the handover interruption time.

\paragraph*{Paper outline} 
Sec.~\ref{sec:related_work} presents the related work. Next, Sec.~\ref{sec:system_model} presents the system model of the \gls{ntn}. \change{Sec.~\ref{sec:protocol} describes the \gls{ntn} \gls{isac} data-frame enabling atmospheric sensing integration.} Sec.~\ref{sec:problem} addresses the \gls{ntn} adaptation, including atmospheric-induced attenuation estimation \change{method and \gls{ra}}. Finally, Sec.~\ref{sec:performance} outlines the performance analysis methodology, Sec.~\ref{sec:results} presents the results, and Sec.~\ref{sec:conclusions} concludes the paper.

\section{Related work}
\label{sec:related_work}
 Most of the work on beam management and allocation in downlink satellite communications considers a single satellite~\cite{Deyi2021, Tang21} and the focus has been on traditional sum rate maximization problems~\cite{Gao21}. Nevertheless, in our preliminary work~\cite{Mayorga2023icc}, we observed that solving the beam management (i.e., switching) and resource allocation problems at each satellite separately leads to a significant degradation either in performance and/or fairness of the solution when compared to a joint optimization that considers all the satellites. 

One of the few works that considers beam management and downlink resource allocation over multiple satellites is~\cite{Lei24}. The latter considers single- beam satellites that can generate different beam patterns and the pattern selection, user-to-satellite association, and resource allocation problems were solved jointly~\cite{Lei24}. While this approach is interesting, it does not fit into the cell classification proposed by the \gls{3gpp} and can be problematic for mobility management. 

Besides solving the problems jointly, it is necessary to consider the cost of adaptation in terms of signaling overhead and mobility management operations, as cell handovers lead to service interruption times. In terrestrial networks, handover interruption times can be as long as $49.5$~ms~\cite{TR36.881}, and handovers in \glspl{ntn} can have an even greater negative impact due to the long propagation delays and the mobility of the infrastructure. Therefore, improvements to the standard handover procedure have been proposed for 5G \gls{ntn}~\cite{3GPPTR38.821}, including the conditional handover and RACH-less handover~\cite{3GPPTR38.821}. Nevertheless, the unavoidable handover of a large number of users is a major challenge in both types of cells: moving and quasi-Earth fixed cells, and further enhancements on the conditional handover are based on triggering the procedure based on the location of the \glspl{ue} or relative antenna gain~\cite{Juan2022}. Despite these enhancements, recent studies by the \gls{3gpp} have indicated that the interruption time for \gls{ntn} is, in the best case, in the order of 2 \glspl{rtt} for downlink communication~\cite{3GPPTR38.821}. 

Furthermore, \gls{ra} must consider the resources dedicated to acquire the desired \gls{csi}. Statistical \gls{csi} can be used for resource allocation with multiple satellites in a decentralized manner. However, the available \gls{ra} schemes based on statistical \gls{csi} require having satellites with strictly different channel qualities (i.e., antenna gains), which is not typical in practical \glspl{ntn}~\cite{Zhao24}. Moreover, their scalability is still a concern, as results are limited to a few satellites only. 
Instead of relying on statistical \gls{csi}, exploiting the predictable dynamics of the \gls{leo} satellites while refining the \gls{csi} by sensing the atmospheric effects is an appealing option. The concept of \emph{opportunistic sensing} of atmospheric phenomena was originally proposed for terrestrial wireless links~\cite{Messer2006:cml-sensing} through the simple, but effective, idea of inverting the \emph{power law formula} linking the rainfall rate with its induced attenuation~\cite{Olsen1978rain}. By measuring attenuation, it is possible to detect a rainfall event and estimate its intensity. However, this requires knowing baseline \gls{rssi} level in \emph{clear sky conditions} and accounting for other factors influencing received power, such as beam misalignment and other weather-induced attenuation~\cite{gritton2005atmospheric, ITU-gasAttenuation, Zinevich2010, Ostrometzky2018baselinelevel}. 

Recently, the research on opportunistic sensing expanded to include satellite communications~\cite{Saggese2022rain}.
The authors of \cite{barthes2013basic} pioneered opportunistic sensing with \gls{geo} satellites, proposing the collection of \gls{rssi} measurements through time and comparing them to a baseline \gls{rssi} value from historical data.
In~\cite{Giannetti2017nefocast}, the authors employ two Kalman filters tracking \emph{slow} and \emph{fast} variations of the \gls{rssi}, respectively, and a rain event is detected when the difference between the outcomes exceed a threshold. 
In~\cite{Arslan2018ber}, the \gls{ber} is constantly measured and logistic regression is used to detect rain periods; the rain intensity is evaluated similarly as before.
\cite{Gharanjik2018ml} proposes an estimation approach based on feature extraction from the \gls{snr} measures, using neural networks trained with labeled data coming from rain gauges and meteorological radar. 

Similar approaches have been proposed for \gls{leo} constellations, addressing the \gls{gsl} variation due to satellite movement. In~\cite{Shen2019tomographic}, the rainfall intensity is estimated through the \gls{snr} measurements at ground receivers: through an iterative least-square-based estimation procedure, the rainfall-induced attenuation is isolated from other factors. Then, a 3D tomographic reconstruction of rainfall intensity is performed by combining data from multiple receivers. In~\cite{Jian2022compressive}, the authors propose a compressive sensing-based 3D rain field tomographic reconstruction algorithm proven to outperform least-square-based approaches. In~\cite{Xian2020:rain-svm-lstm}, the authors use a support vector machine algorithm to detect rain events, while performing a time-series analysis with long short-term memory neural networks to estimate the baseline \gls{rssi} value in clear sky condition. 
Nevertheless, none of these studies explore the coexistence of sensing and data communication in \gls{leo} \glspl{ntn}.



\section{System model}
\label{sec:system_model}

\subsection{Satellite constellation and cells on ground}
\label{sec:sat-cells}
We consider a \gls{dl} scenario in which a heterogeneous \gls{leo} satellite constellation with $S$ satellites that serve users on the ground through a direct satellite-to-user access link~\cite{TS38.300,Guidotti2018}. We limit our analysis to a fixed geographical region of the Earth and consider the quasi-Earth fixed cell scenario described by the \gls{3gpp}~\cite{3GPPTR38.821, TS38.300}. Therefore, the users within the region are aggregated into $C$ fixed and uniformly distributed geographical cells as shown in Fig.~\ref{fig:users_map}. The set of cells is denoted as $\set{C}$ and the total number of users in the cell is $M_c^\text{max}$. A fraction $\mu_c\in\left[0,1\right]$ of the users in cell $c\in\mc{C}$ are actively receiving data from the satellite network, such that the number of active users in cell $c\in\set{C}$ is $M_c=\left\lceil\mu_c M_c^\text{max}\right\rceil$. The remaining $M_c^\text{max}-M_c$ users are either inactive or communicating through the terrestrial infrastructure.

The set of satellites in the heterogeneous constellation is denoted as $\set{S}$, which is further divided into subsets of satellites $\set{S}_0, \set{S}_1, \dotsc$, which denote different \emph{orbital shells}, with each shell having different characteristics. Each satellite $s\in\set{S}_i$ operates at a given frequency band with carrier frequency $f_s$, bandwidth $B_s$, antenna gain $G_s$, are deployed at an altitude $h_s$, which is equal for all $s\in\set{S}_i$, and has a footprint covering up to $C_s$ cells.
At least one orbital shell operates at a carrier frequency in the K-band or higher, as operating in such high frequencies enables atmospheric sensing~\cite{Giannetti2017nefocast}. 
\change{Furthermore, the satellites in each orbital plane are evenly spaced across the orbit, have the same inclination angle, and have the same orbital velocity} 
\begin{equation}
v_s = \sqrt{\frac{\mathrm{G}M_E}{R_E+h_s}},
\end{equation}
where $\mathrm{G}$ is the universal gravitational constant, $M_E$ is the Earth's mass, and $R_E$ is the Earth's radius~\cite{SatBook}.
Finally, each satellite is equipped with a precise multi-beam system with $N_B$ independent beams that can be steered toward the intended cells throughout the satellite movement. The maximum communication range for satellite $s$ is defined from its altitude $h_s$ and the minimum elevation angle $\eta$ as
\begin{equation}
d_s(\eta)=\sqrt{R_E^2\sin^2\left(\eta\right)+2R_Eh_s+h_s^2}-R_E\sin\left(\eta\right).
\label{eq:range}
\end{equation}

\begin{table}[t]
    \centering
    \renewcommand{\arraystretch}{1.1}
    \caption{Important parameters and variables}
    \begin{tabularx}{\columnwidth}{@{}lXc@{}}
    \toprule   
    \multicolumn{2}{l}{\textbf{Parameter}} & \textbf{Symbol} \\
    \midrule
    \multicolumn{2}{l}{\textbf{Ground segment}}\\
        &Number of cells in the area & $C$\\
        &Set of cells in the area & $\set{C}$\\
        &Population for cell $c$ & $M_c^\text{max}$\\
        &Fraction of active users per in cell $c$ & $\mu_c$\\
        &Number of active users in cell $c$ & $M_c$\\
    \multicolumn{2}{l}{\textbf{Space segment}}\\
        &No. of satellites in the area at frame $k$  & $S(k)$\\
        &Set of satellites in the area at frame $k$ & $\set{S}(k)$\\
        & Number of beams per satellite & $N_B$\\        
        &Altitude of satellite $s$& $h_s$\\
        &Bandwidth per beam for satellite $s$ & $B_s$\\
        &Carrier frequency for satellite $s$ & $f_s$\\
        &No. of cells covered by satellite $s$ footprint & $C_s$ \\
        &Set of cells served by satellite $s$ in frame $k$ & $\mathcal{E}_s(k)$ \\
      \multicolumn{2}{l}{\textbf{Timing}}\\
      & Duration of a system frame & $T_F$\\
      & Duration of an \gls{ofdma} frame & $T$\\
      & No. of \gls{ofdma} frames per system frame &$N_T$\\      
      & No. of \gls{ofdma} frames for communication per system frame& $N_C$ \\
      & No. of \gls{ofdma} frames for sensing per system frame& $N_S$ \\
      & Frame index & $k$\\           
    \multicolumn{2}{l}{\textbf{Ground-to-satellite links}}\\
        & Euclidean distance for link $s-c$ in frame $k$ & $d_{s,c}(k)$\\
        & Free-space path loss for link $s-c$ in frame $k$ & $\set{L}_{s,c}(k)$\\
        & Atmospheric attenuation for link $s-c$ in frame $k$ &$A_{s,c}(k)$\\
        & \gls{snr} for link $s-c$ in frame $k$&  $\gamma_{s,c}(k)$\\
        \multicolumn{2}{l}{\textbf{\Gls{snr} estimation}}\\
        & Length of the pilot signals (symbols) & $L_p$\\
        &Estimated \gls{snr} link $s-c$ in frame $k$ & $\hat{\gamma}_{s,c}(k)$\\
         \bottomrule
    \end{tabularx}
    \label{tab:symbols}
\end{table}

\subsection{Discrete timing model}
\label{sec:frame1}
We consider a \change{sensing-assisted satellite communications system whose operation is divided into system frames of duration $T_F$. Under this frame structure, the network dynamics are modeled following a discrete timing model, whose minimum time unit is the system frame $T_F$, referred to as frame for simplicity and indexed by $k\in\{0, 1, \dots\}$. In our discrete timing model, the changes in the propagation environment due to the satellite movement (described previously) and atmospheric dynamics (described in Section~\ref{sec:rain}) occur at the beginning of each system frame. Consequently, the propagation environment can be considered static within each frame.}  Within each frame, the \gls{ntn} is adapted to provide a certain \gls{qos} to the \glspl{ue}. In particular, each satellite $s$ among the subset of satellites covering the area at frame $k$, namely $\mathcal{S}(k)\subseteq \mathcal{S}$, having cardinality $S(k)$, is assigned to cover a subset of the cells in the area. 

To comply with the \gls{3gpp} standards for \gls{ntn} 5G \gls{nr}~\cite{3GPPTR38.811, 3GPPTR38.821}, the system frame $T_F$ is a multiple of the duration of an \gls{ofdma} frame $T$, i.e., $T_F = N_T T$, \change{the latter representing the minimum amount of time that satellite $s$ can allocate bandwidth resources to serve the \glspl{ue} in cell $c$.} Furthermore, we assume that every satellite can steer each of its $N_B$ beams once per \gls{ofdma} frame, to align them to specific ground cells. Among the $N_T$ \gls{ofdma} frames, $N_C \le N_T$ are considered available for communication, while \change{$N_S\le N_T$} are available for sensing. \change{Therefore, the network resources are allocated across time (OFDMA frames), frequency (bandwidth), and space (beams) domains.}
Finally, we assume that the \gls{ofdma} frames used for data transmission, along with the multiple resource blocks within each \gls{ofdma} frame, belonging to the same cell and satellite beam pair are allocated uniformly to the users within the cell.
The details on the \gls{isac} frame structure are given in Sec.~\ref{sec:protocol}.

\begin{figure}[t]
\flushright
\subfloat[]{\includegraphics[width=0.9\columnwidth]{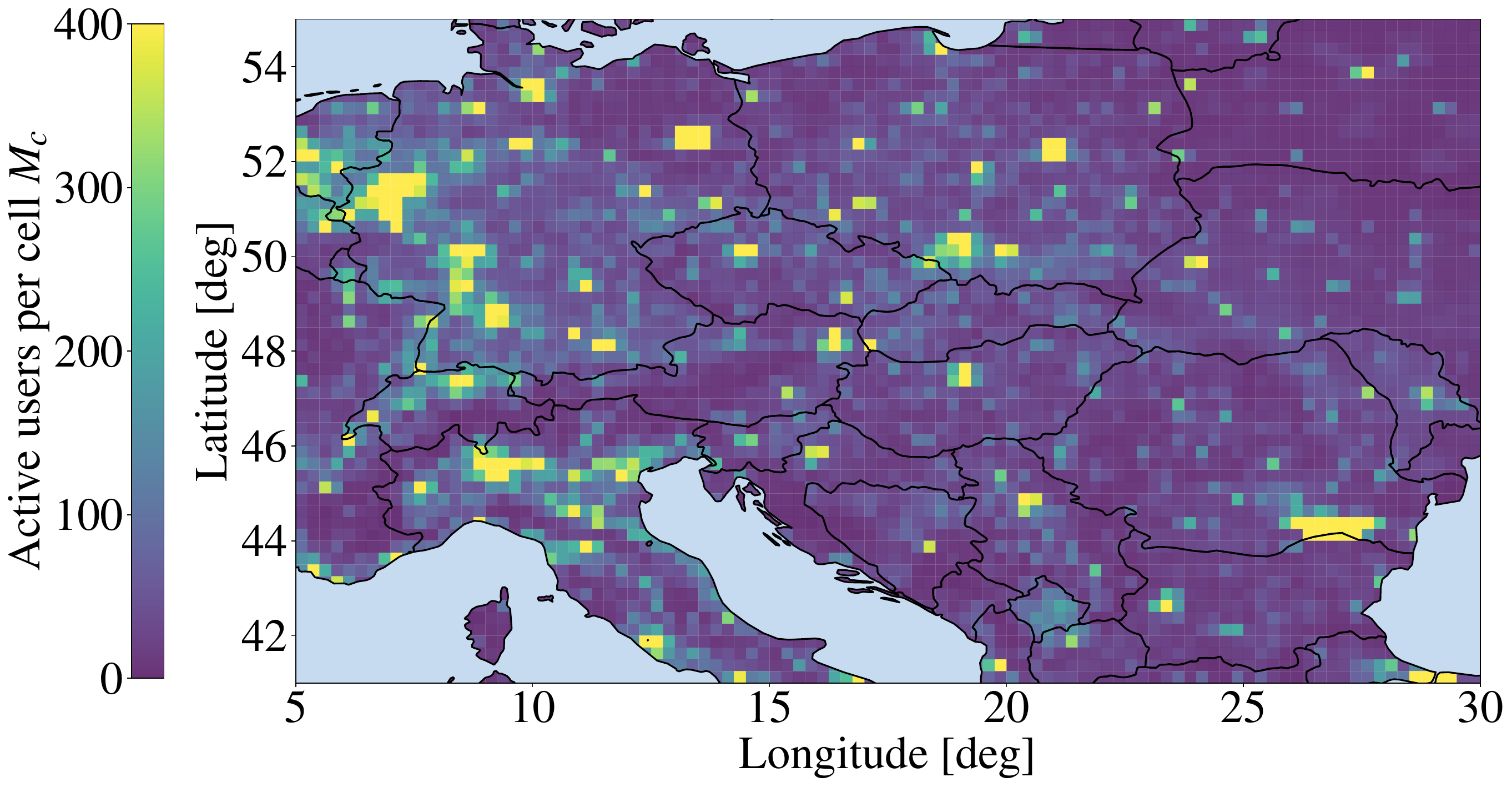}\hspace{0.05\columnwidth}\vspace{-0.5em}\label{fig:users_map}}
\\
\subfloat[]{\includegraphics[width=0.885\columnwidth]{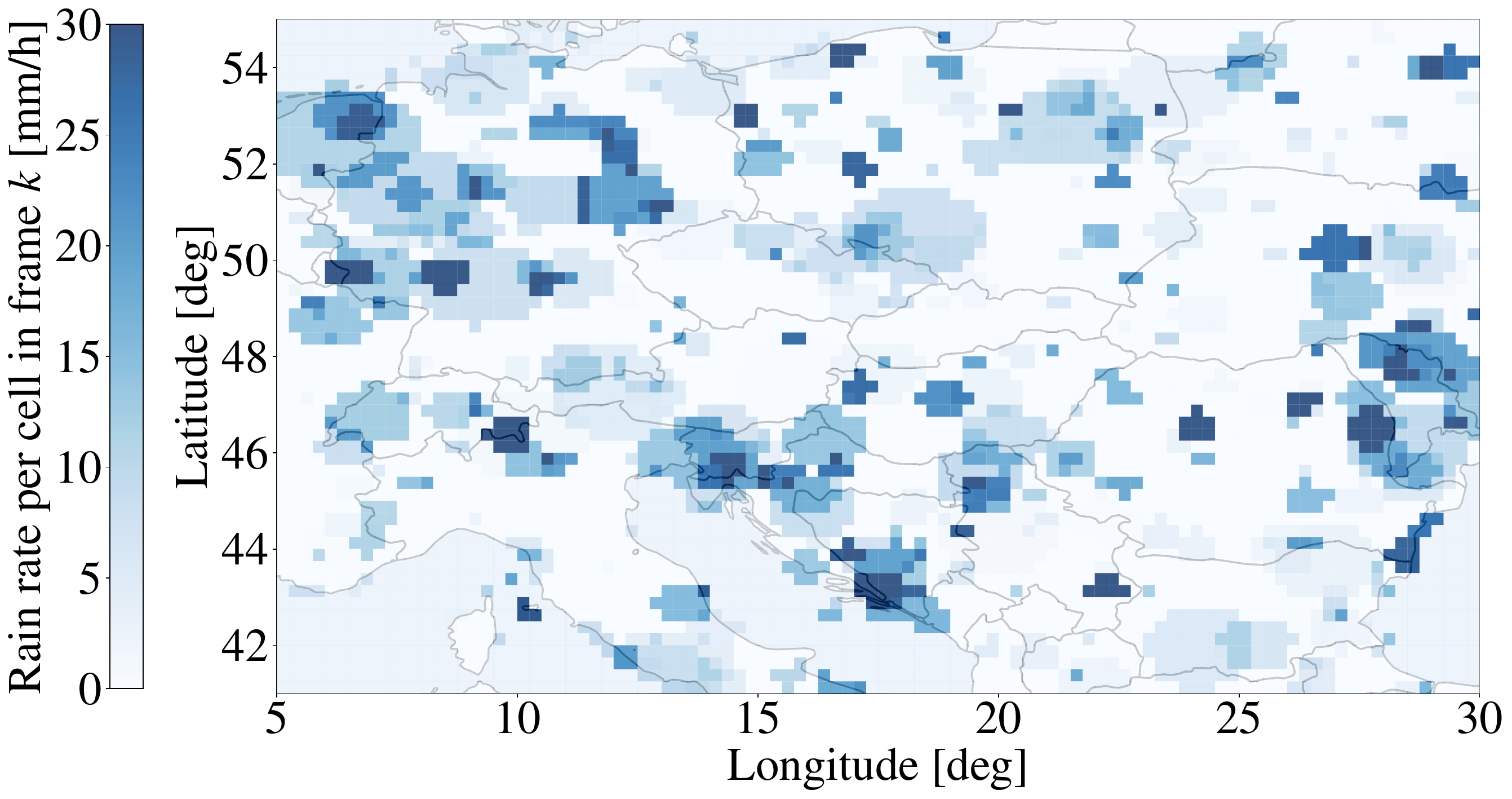}\hspace{0.05\columnwidth}\vspace{-0.5em}\label{fig:rain_map}}
\caption{Map of central Europe illustrating (a) the active users per-cell as in Sec.~\ref{sec:sat-cells} (derived from the per-cell population~\cite{CIESIN}) and (b) the rain rate with the clustered rain model of Sec.~\ref{sec:rain}.}
\label{fig:maps}
\end{figure}

\subsection{\change{Ground-to-Satellite} links}
\label{sec:gsl}
\Gls{dl} communication takes place through direct satellite-to-user \gls{awgn} channel links, where the attenuation \change{of the transmitted signal} is given by the free-space path loss and the atmospheric conditions. 
 
To achieve reliable communication to all the users within the area, we consider the problem of beam-cell allocation using the worst-case performance for a user in each cell $c\in\set{C}$. Therefore, the free-space path loss between cell $c$ and a satellite $s\in\set{S}$ at the $k$-th frame is given by 
\begin{equation}
   \set{L}_{s,c}(k) = \left(4 \pi d_{s,c}(k) f_s\right)^2 v_c^{-2} , 
   \label{eq:loss}
\end{equation}
where $d_{s,c}(k)$ is the maximum distance between satellite $s$ and any point in cell $c$ at frame $k$, and $v_c$ is the speed of light. Being the cell size much smaller than the altitude of the satellites, only a small difference between the minimum and maximum attenuation within a cell is expected. In our experiments, this value was always below $0.3$~dB.

Next, let $P_s$ be the transmission power of satellite $s$, $G_{c,s}$ be the antenna gain of the users in cell $c$ towards satellite $s$, $\sigma^2$ denote the noise power at the users' receivers, and $\ell$ be the pointing loss. Moreover, let $A_{s,c}(k)$ be the atmospheric attenuation generated by the presence of water particles affecting the link between cell $c$ and satellite $s$ at frame $k$. 
Thus, the \gls{snr} from a cell $c$ to satellite $s$ at frame $k$ is
\begin{equation}
    \gamma_{s,c}(k)=\frac{P_s G_{c,s} G_{s}}{\set{L}_{s,c}(k) A_{s,c}(k)\,\ell\,\sigma^2}.
    \label{eq:snr}
\end{equation}
The data rate for frame $k$ is selected based on the \gls{snr} $\gamma_{s,c}(k)$ at the beginning of the frame and remains fixed until the next frame, considering an \gls{awgn} channel and a sufficiently large blocklength such that the finite blocklength regime effects are negligible. Furthermore, to avoid link outages within the individual frames, we limit the communication to satellite-cell pairs that remain within communication range throughout the whole duration of the frame, i.e., those fulfilling $d_{s,c}^\mathrm{max}(k) = \max\{d_{s,c}(k), d_{s,c}(k+1)\}\leq d_s(\eta)$. Accordingly, the maximum achievable data rate for reliable communication for any user in a satellite-cell pair $(s,c)$ in frame $k$ is 
\begin{equation}   
    \rho_{s,c}(k)  \hspace{-.7mm}= \hspace{-.7mm}
    \begin{cases} 
    B_s\log_2\left(1+\gamma_{s,c}(k)\right), &\text{if } d_{s,c}^\mathrm{max}(k) \leq d_s(\eta),\\
    0, &\text{otherwise}.
    \end{cases}
    \label{eq:achievable-rate}
\end{equation}

\subsection{Atmospheric attenuation} 
\label{sec:rain}
In this paper, we consider that sensing is performed between the satellites and anchor nodes with known positions. These can either be dedicated gateways or mobile users. Building on this, and since the distance between the satellites and anchor nodes can be accurately calculated (and predicted) from the ephemeris, the attenuation due to atmospheric phenomena is the only unknown in~\eqref{eq:snr} that prevents a perfect rate selection.
This atmospheric attenuation at frame $k$, $A_{s,c}(k)$, depends on three main effects related to the presence of water molecules~\cite{Ostrometzky2018baselinelevel}: $i$) attenuation due to wet antenna effect generated by water accumulated on the antennas’ radomes or reflectors, which is related to the amount of rain but largely independent of the instantaneous rain rate~\cite{Kharadly2001wetantenna}; $ii$) attenuation due to rainfall impacting on the \gls{gsl}; and, $iii$) the attenuation due to other-than-rain phenomena, which comprises fog, water vapor, humidity, and atmospheric gasses, and depends mainly on the frequency of the signal~\cite{Ostrometzky2018baselinelevel, ITU-gasAttenuation}.
In the remainder of the paper, we consider the simplifying assumption that the attenuation due to wet antenna effects and other-than-rain phenomena are negligible\footnote{We remark that many solutions are available to retrieve the rainfall attenuation from the total atmospheric attenuation, e.g.,~\cite{Zinevich2010, Giannetti2017nefocast, Jian2022compressive}.}, and, thus, the atmospheric attenuation $A_{s,c}(k)$ in~\eqref{eq:snr} is only caused by rain phenomena. Accordingly, we assume that $A_{s,c}(k)=1$ if the \gls{gsl} experiences no rain. Otherwise, the atmospheric attenuation in the presence of uniform rain within a cell\footnote{We implicitly assume that the rainfall is approximately constant in a cell, and that the slanted path of the \gls{gsl} does not pass over other cells below the rain height. These two assumption are not required for the proposed design to work, but are used to simplify the implementation of the simulations.}, given in dB, can be closely approximated by the \emph{power law}, stating~\cite{Olsen1978rain} 
\begin{equation} \label{eq:atm-loss}
    A_{s,c}^\text{dB}(k) = \mu_s \, \left[\varrho_c\left(k\right)\right]^{\omega_s} \tilde{d}_{s,c}(k), \,\, \text{[dB]}
\end{equation}
where $\mu_s$ and $\omega_s$ are empirical coefficients which depends on the carrier frequency $f_s$ and on the polarization of the signal transmitted by the satellite $s$~\cite{ITU-rain, Zhao2001}; $\varrho_c(k)$ is the rainfall intensity in mm/h for cell $c$; $\tilde{d}_{s,c}(k)$ is the distance along the line-of sight path between the ground terminal in cell $c$ and the rain height, i.e., the 0$^\circ$ C isotherm height in km~\cite{ITU-rainheight}.

We consider a clustered model for the rain, adapted from~\cite{Cow95}, where the number of \emph{rain cells} $N_\text{rain}$ and their centers $z_r$ for $r\in\left\{1,2,\dotsc,N_\text{rain}\right\}$ are defined by a two-dimensional \gls{ppp} with intensity $\lambda_\text{rain}$ cells/\SI{}{\kilo\metre\squared}. The radius, intensity, duration, and inter-arrival times of the individual rain cells are exponentially distributed with parameters $d_\text{rain}$~km, $\overline{\varrho}$~mm/h, $\varepsilon$~h, and $\beta$~h, respectively.

In our discrete time model, illustrated in~\ref{fig:rain_model}, the changes in the atmospheric attenuation due to rain occur at the beginning of a frame $k$ and remain fixed until the beginning of frame $k+1$, analogous to block fading models. Thus, we adapt the original continuous time rain model~\cite{Cow95} to a discrete timing model, where the on-off process of the rain cells is a two-state \gls{dtmc}. Let $\alpha_r(k)\in\{0,1\}$ be an indicator variable that takes the value of $1$ if rain cell $r$ is active at frame $k$ and $0$ otherwise. Consequently, the probability that an active rain cell at time $k-1$ becomes inactive  at frame $k$ is 
\begin{equation}
p_\text{off}=  \Pr\left(\alpha_r(k)=0\!\mid\! \alpha_r(k-1)=1\right) = 1-e^{-T_F/\varepsilon}.
\end{equation}
Analogously, the probability that an inactive rain cell at frame $k-1$ becomes active at frame $k$ is 
\begin{equation}
    p_\text{on}=\Pr\left(\alpha_r(k)=1\!\mid\! \alpha_r(k-1)=0\right) =\! 1-e^{-T_F/\beta}.
\end{equation}
Building on these, it is easy to calculate the steady state probabilities of a rain cell as 
\begin{equation}
\begin{aligned}
    \pi_\text{on} &= \frac{p_\text{on}}{p_\text{on}+p_\text{off}}=\frac{1-e^{-T_F/\beta}}{2-e^{-T_F/\beta}-e^{-T_F/\varepsilon}}=1-\pi_\text{off}. 
\end{aligned}    
\end{equation}

Next, let $d_{r,c}$ be the distance between rain cell $r$ and cell $c$, $\varrho_r$ be the rain rate at rain cell $r$, and $\phi_r$ be the diameter of rain cell $r$.
The instantaneous rain rate for an active rain cell $r$ is exponentially distributed as
    $\Pr\left(\varrho_r\leq \varrho\right)= 1- e^{\varrho/\overline{\varrho}}$. Finally, we define $\mathcal{R}_c=\left\{r\in\left[N_\text{rain}\right]:d_{r,c}\leq \phi_r/2\right\}$ to be the set of rain cells that have cell $c$ within their radii, and calculate the intensity of the rain for a given cell $r_c$ as~\cite{Cow95}
\begin{equation}
    \varrho_c(k) = \sum_{r\in\mathcal{R}_c} \alpha_r(k)\varrho_r.
\end{equation}

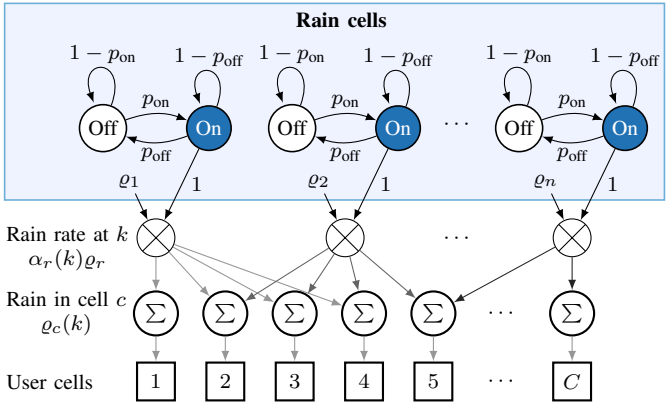
\begin{figure}[t]
    \centering
    \tikzset{usercell/.style={rectangle, fill=none, draw=Greys-M, thick, minimum size=0.5cm},
rain_cell/.style={minimum size=0.5cm, circle, inner sep=2pt, draw=black, semithick},
cross/.style={path picture={ 
  \draw[black,-]
(path picture bounding box.south east) -- (path picture bounding box.north west) (path picture bounding box.south west) -- (path picture bounding box.north east);
}}}
\begin{tikzpicture}[->, >=latex, font=\footnotesize]
\pgfmathsetmacro{\hrain}{-1.5}
\pgfmathsetmacro{\xm}{2.5}
\pgfmathsetmacro{\xml}{3.2}
\pgfmathsetmacro{\hcells}{-3.4}
\pgfmathsetmacro{\hsum}{-2.5}
\pgfmathsetmacro{\xcelln}{7}
\pgfmathsetmacro{\xcells}{(\xml-1)*\xm/(\xcelln-1)}
\pgfmathsetmacro{\xlabel}{-2.1}
\pgfmathsetmacro{\xmsep}{0.7}
\pgfmathsetmacro{\toptop}{1.6}
\pgfmathsetmacro{\hmarkov}{-0.05}

\filldraw[fill=Blues-B, draw=Blues-G, semithick] (\xlabel+0.1,\toptop) rectangle (6.72,\hrain+0.5); 
\node at({(7+\xlabel)/2},\toptop)[anchor=north] {\textbf{Rain cells}};

 \node [draw,circle, cross, minimum size=0.5cm](1) at (0,\hrain){}; 

 \node [draw,circle, cross,minimum size=0.5cm](2) at (\xm,\hrain){}; 
  \node [draw,circle, cross,minimum size=0.5cm](n) at (\xml*\xm-\xm,\hrain){}; 
 
\node (3) at(\xm*1.6,\hrain) {$\cdots$};

\foreach \i in {1,2}{
\node[rain_cell]   (A\i) at (\xm*\i-\xmsep-\xm,\hmarkov)[fill=white]   {Off} ;
\node[rain_cell]    (B\i) at (\xm*\i+\xmsep-\xm,\hmarkov)[fill=Blues-J]   {\textcolor{white}{On}};
\path
 (A\i) edge[bend left=20]     node[above, inner sep=2pt]{$p_\text{on}$}         (B\i)
 (A\i) edge[in=110,out=70,loop]	    node[pos=0.5,above, inner sep=2pt]{$1-p_\text{on}$}         (A\i)
 (B\i) edge[bend left=20]     node[below, inner sep=2pt]{$p_\text{off}$}         (A\i)
 (B\i) edge[in=110,out=70,loop]	    node[pos=0.5,above, inner sep=2pt]{$1-p_\text{off}$}         (B\i)
;
\draw[->] (B\i)--(\i)node[midway,right]{$1$};
\node[ inner sep=2pt] (rho\i)  at(\xm*\i-\xmsep/2-\xm,\hrain/2){$\varrho_\i$};
\draw[->] (rho\i)--(\i);
}
\node[rain_cell]   (A3) at (\xm*\xml-\xmsep-\xm,\hmarkov)[fill=white]   {Off} ;
\node[rain_cell]    (B3) at (\xm*\xml+\xmsep-\xm,\hmarkov)[fill=Blues-J]   {\textcolor{white}{On}};
\path
 (A3) edge[bend left=20]     node[above, inner sep=2pt]{$p_\text{on}$}         (B3)
 (A3) edge[in=110,out=70,loop]	    node[pos=0.5,above, inner sep=2pt]{$1-p_\text{on}$}         (A3)
 (B3) edge[bend left=20]     node[below, inner sep=2pt]{$p_\text{off}$}         (A3)
 (B3) edge[in=110,out=70,loop]	    node[pos=0.5,above, inner sep=2pt]{$1-p_\text{off}$}         (B3)
;
\node at(\xm*1.6,0) {$\cdots$};

\draw[->] (B3)--(n)node[midway,right]{$1$};
\node[inner sep=2pt] (rhon)  at(\xm*\xml-\xmsep/2-\xm,\hrain/2){$\varrho_n$};
\draw[->] (rhon)--(n);

\node at(\xlabel,\hrain)[anchor=west, align=center, yshift=-4pt] {Rain rate at $k$\\$\alpha_r(k)\varrho_r$};

\node at(\xlabel,\hsum)[anchor=west, align=center] {Rain in cell $c$ \\$\varrho_c(k)$};
\node at(\xlabel,\hcells)[anchor=west] {User cells};
\foreach \i in {1,2,...,5}{
\node[circle,draw=black, thick, inner sep=2]    (s\i) at (\i*\xcells-\xcells,\hsum) {\scriptsize $\sum$};
\node[usercell]    (c\i) at (\i*\xcells-\xcells,\hcells) {$\i$};
\draw[->, Greys-G] (s\i)--(c\i.north);}
\node at(5*\xcells,\hcells) {$\cdots$};
\node at(5*\xcells,\hsum) {$\cdots$};

\node[circle,draw=black, thick, inner sep=2]    (s\xcelln) at (\xcelln*\xcells-\xcells,\hsum) {\scriptsize $\sum$};
\node[usercell]    (c\xcelln) at (\xcelln*\xcells-\xcells,\hcells) {$C$};
\draw[->, Greys-G] (s\xcelln)--(c\xcelln.north);

\foreach \i in {1,2,3,4}{
\draw[->, Greys-G] (1)--(s\i);}

\foreach \i in {2,3,4,5}{
\draw[->,Greys-I] (2)--(s\i);}

\foreach \i in {5,7}{
\draw[->,Greys-K] (n)--(s\i);}
\end{tikzpicture}
    \vspace{-12pt}
    \caption{The rain rate at each user cell is the sum of rates of the active rain cells within range. The on-off process for the rain cells is a two-state \gls{dtmc}.}
    \vspace{-12pt}
    \label{fig:rain_model}
\end{figure}



\section{Non-terrestrial ISAC data-frame}
\label{sec:protocol}
\change{This section explores the structure of the proposed \gls{isac} data-frame, which facilitates the integration of sensing in \glspl{ntn} and serves as a foundation for developing atmospheric sensing and \gls{ra} algorithms. The data-frame consists of sub-frames dedicated to atmospheric sensing, \gls{ra}, and data communication. Generally, resources in the spatial, frequency, time, and computational domains are reserved for these sub-frames, and the same resources may or may not be shared with others, depending on system requirements. Given the \emph{causality} between sensing, \gls{ra}, and communication, we focus on analyzing the time-domain constraints that enable real-time network adaptation to the environmental dynamics.}

\change{The time-domain} design of a sensing and communication data frame depends on the time horizons of the two tasks. 
As mentioned in Sec.~\ref{sec:frame1}, the frame duration must be sufficiently short so that the \gls{snr} remains approximately constant within the frame, avoiding degradation of the expected communication performance. 
From~\eqref{eq:snr}, this condition holds when the large-scale channel conditions~\eqref{eq:loss}, determined by the satellite-to-cell distances and the atmospheric attenuation~\eqref{eq:atm-loss} remain approximately constant throughout the frame. 
In general, it is reasonable to assume that rain events change with an average time horizon of minutes or even hours~\cite{Giannetti2017nefocast, Cow95}. Thus, the variation of the attenuation caused by atmospheric effects is slower than the variation of the free space path loss. 
Therefore, $T_F$ must be sufficiently short so that all 
$\mc{L}_{s,c}$ remain approximately constant, which can be calculated from the cell distribution of the area of interest and the constellation parameters. 
Nevertheless, an extremely short frame might introduce an excessive amount of overhead due to highly frequent sensing and \gls{ra} \change{sub-frames consuming resources needed for communication}.

An important factor to dimension the \change{duration} of the sensing and \gls{ra} sub-frames is the propagation time. In particular, the \change{time employed by} these sub-frames must be sufficiently long to accommodate the maximum propagation time for all the cells within the communication range of the satellites, i.e., 
\begin{equation}
    t(\eta)=\max_{s}\frac{d_s(\eta)}{v_c}=\max_{s,c,k} \frac{d_{s,c}(k)}{v_c}.
    \label{eq:prop_time}
\end{equation}

According to the previous considerations, and following the assumptions of Sec.~\ref{sec:frame1}, we propose the \gls{isac} data frame shown in Fig.~\ref{fig:time-diagram}. Each system frame $k$ has duration $T_F = N_T T$ and addresses atmospheric sensing, \gls{ra} and data communication, as described in the following subsections.


\begin{figure}
    \centering
    \begin{subfigure}{\columnwidth}
        \centering
         \includegraphics[width=\columnwidth]{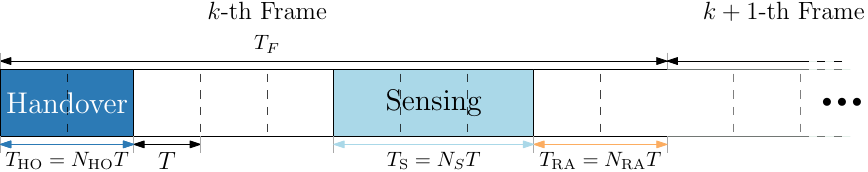}
         \caption{Overall \gls{isac} frame.}
        \label{fig:time-diagram}         
    \end{subfigure}
    \begin{subfigure}{\columnwidth}
        \centering
        \includegraphics[height=1.7cm]{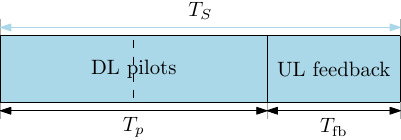}    
        \caption{Sensing sub-frame.}
        \label{fig:sensing-frame}
    \end{subfigure}
    \\[1em]
    \begin{subfigure}{\columnwidth}
        \centering
        \includegraphics[height=1.7cm]{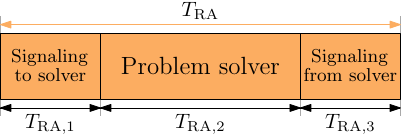}    
        \caption{\gls{ra} sub-frame.}
        \label{fig:ra-frame}
    \end{subfigure}
    \caption{Visualization of a generic $k$-th \gls{isac} frame for a $(s,c)$ pair, comprising communication, sensing sub-frame, and \gls{ra} process, with $T_F = 10 T$, $N_S = 3$ and $N_\mathrm{HO} = 2 $. The handover on frame $k+1$-th is missing to show a situation where the $(s,c)$ pair does not change from $k$ to $k+1$.}
    \label{fig:frame}
\end{figure}

\subsection{Sensing sub-frame}
\label{sec:sensing-frame}
A \emph{sensing sub-frame} of duration $T_S = N_S T \le T_F$ is considered for atmospheric sensing. 
\change{Satellites operating in the K-band or higher perform this task by transmitting pilot sequences to cells within their footprint, while lower-band satellites, unaffected by water particles, are not involved.}

\change{In this paper, we assume that an anchor node placed in each cell -- e.g., a dedicated ground station, gateway, or a mobile user with a known position -- estimates the \gls{snr} as a metric related to \gls{gsl} conditions and reports it to the associated satellite\footnote{\change{Similar considerations can be drawn for multiple ground nodes performing the sensing process, eventually using different metrics, e.g.,~\gls{rssi}~\cite{Giannetti2017nefocast, Saggese2022rain}. Here, we focus on a anchor node per cell for simplicity of presentation.}}. Using this report and its known position, the satellite can estimate the atmospheric attenuation, e.g., via the method of Sec.~\ref{sec:snr}. The sensing process must terminate before the \gls{ra} process starts.} Accordingly, the sensing sub-frame is divided into two parts, shown in Fig.~\ref{fig:sensing-frame}: \gls{dl} pilots and \gls{ul} feedback.

\change{To analyze the sensing sub-frame duration, let $C_s$ be the upper bound on the number of cells sensed by any satellite $s$, while $B_s$ and $N_B$ the bandwidth (frequency domain) and the number of beams (spatial domain) available for sensing. The following considerations apply also when subsets of frequency and spatial resources are reserved for sensing, and the remaining resources are simultaneously employed for communication.}

\change{During the \gls{dl} pilot phase, satellites transmit pilot sequences of length $L_p$ to the anchor nodes, which then estimate the \gls{kpi}.} 
The time needed to transmit a pilot to an anchor node is  $L_p/B_s$, assuming the symbol rate is equal to the available bandwidth. Since $C_s>N_B$, the pilots are sent in parallel to $N_B$ anchor nodes. Then, the beams are switched towards other anchor nodes a total of $\left\lceil C_s/N_B \right\rceil$ times, being $\lceil \cdot \rceil$ the ceiling operator. Finally, since pilot transmissions start at the beginning of the sensing sub-frame and these are affected by the propagation time, which is upper bounded by $t(\eta)$, the duration of the \gls{dl} pilot transmission phase is
\begin{equation} \label{eq:pilot}
   T_p = T\left\lceil\left(t(\eta)+\left\lceil \frac{C_s}{N_B} \right\rceil \cdot \frac{L_p}{B_s}\right)\middle/T\right\rceil.
\end{equation}

During the \gls{ul} feedback portion of the sub-frame, the anchor node reports the estimated \gls{snr} to the satellite. The feedback duration depends on the number of symbols to be transmitted $L_\text{fb}$ and the number of times that the $N_B$ beams must be switched to receive the feedback from the $C_s$ anchor nodes, similar to the \gls{dl} pilot case.
Accordingly, the duration of \gls{ul} feedback part can be formulated is, in the worst case,
\begin{equation} \label{eq:ul-feedback}
   T_\text{fb} = T\left\lceil\left(\left\lceil \frac{C_s}{N_B} \right\rceil \cdot \frac{L_\text{fb}}{B_s}+ t(\eta)\right)\middle/T\right\rceil.
\end{equation}
Finally, the length of the sensing sub-frame is calculated as
\begin{equation} \label{eq:sensing-time}
T_S = T_p + T_\text{fb},
\end{equation}
which comprises $N_S=T_S/T$ \gls{ofdma} frames.

\begin{figure}[t]
    \centering
    \includegraphics[width=0.7\columnwidth]{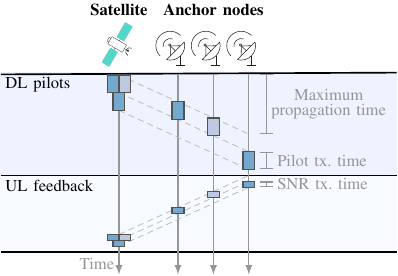}
    \caption{Exemplary time diagram for the sensing sub-frame with $C_s=3$ anchor nodes and $N_B=2$.}
    \vspace{-12pt}
    \label{fig:timediag_sensing}
\end{figure}

\subsection{\gls{ra} sub-frame}
\label{sec:ra-frame}

Based on the sensing outcome, the \gls{ra} optimization problem is solved in a \emph{centralized manner} by a single decision-maker -- e.g., a ground station -- and its solution is applied for the whole duration of the next frame. 
The complete \gls{ra} process has duration $T_\mathrm{RA} = N_\mathrm{RA} T$ and it is performed in parallel to data transmission, allowing the satellites to communicate with the \glspl{ue} while solving the optimization problem. 
The \gls{ra} process is comprised by three parts, as shown in Fig.~\ref{fig:ra-frame}: signaling to the decision-maker, problem solver, and signaling from the decision-maker.
In the first and third parts, the data in input and output to and from the \gls{ra} are forwarded to the decision-maker and to the satellites covering the area, respectively. These parts last for $T_{\mathrm{RA},1}$ and $T_{\mathrm{RA},3}$ seconds, whose values depend on the time needed to route the traffic between the satellites and the decision-maker. 
During the problem solver part, the decision-maker uses its computational resources to compute the solution of the \gls{ra} problem (cf. Sec.~\ref{sec:optimization}). This part last for $T_{\mathrm{RA},2}$ seconds, usually being the dominant term due to the complexity of the \gls{ra} problem. 
The \gls{ra} sub-frame must occur \emph{after} the sensing sub-frame to gather inputs for optimization, and handovers (detailed in Sec.~\ref{sec:comm-frame}). Moreover, while $T_\text{RA}$ depends on the algorithm complexity and the available hardware, the following condition must be fulfilled to achieve a real-time operation of the system.
\begin{equation} \label{eq:ra-time}
    T_\mathrm{RA} \le T_F - T_S - T_\mathrm{HO}.
\end{equation}

\subsection{Communication sub-frame}
\label{sec:comm-frame}

\change{The communication sub-frame spans the time, frequency and spatial domain resources not employed for sensing.}
At the beginning of each frame, a handover disconnection time due to switching between different satellite-cell associations is considered. During handover, the satellites steer their beam to transmit data and control signals to the cells to be served resulting from the \gls{ra} process. 
According to the \gls{ofdma} frame structure, the handover disconnection time is defined as $T_\mathrm{HO} = N_\mathrm{HO} T$, and occurs on a satellite-cell pair frame only if the satellite has switched covered cells from the previous system frame, as accounted by~\eqref{eq:cost} in the \gls{ra} problem.
Calculating the handover duration depends on the specific signaling needed to perform the handover and the maximum propagation time from ground to satellite $t(\eta)$. 
Thus, we will resort to test the system considering $T_\mathrm{HO}$ as a function of the round-trip time $2t(\eta)$ and the type of handover mechanism implemented. Namely,
\begin{equation}
    T_\text{HO} \geq T\left\lceil \left(N_\text{RTT}2\,t(\eta)\right)/T \right\rceil,
\end{equation}
where $N_\text{RTT}= 2$ for the traditional handover mechanism. 
Finally, data transmission occurs \change{on the resources} not used for sensing or communication overhead due to handover.

\section{Problem formulation}
\label{sec:problem}
To address the \gls{ntn} adaptation, we consider a joint satellite-to-cell matching and \gls{ra} problem, where the satellites must allocate resources to the cells in the area of interest to achieve an efficient and fair \gls{qos}. \change{We consider that sensing and communication tasks are orthogonal in the time domain. Therefore, no communication occurs during the $N_S$ \gls{ofdma} frames composing the sensing sub-frame (cf. Sec.~\ref{sec:sensing-frame}) and vice-versa. Moreover, spatial and frequency domain resources are fully available for sensing \emph{and} for communication tasks in their respective sub-frames. Therefore, we focus on allocating a portion of the $N_C = N_T - N_S$ remaining \gls{ofdma} frames and the $N_B$ beams to each satellite-cell pair, assuming the whole bandwidth $B_s$ is employed.}
The optimization objective is to maximize the fairness among the per-user data rates across a pre-defined area. Nevertheless, the optimal resource allocation problem requires to have perfect knowledge about the attenuation of all the satellite-to-cell links. Even though the satellite-to-cell distances can be accurately calculated from the ephemeris, there is uncertainty in the attenuation of the links due to unknown atmospheric conditions, namely, rain. Therefore, we decompose the problem into 1) \gls{snr} estimation through \gls{mle} \change{performed during the sensing sub-frame} and 2) optimal resource allocation for communication \change{performed during the \gls{ra} sub-frame}. 

\subsection{\gls{mle} for \gls{snr} and the \gls{crb}}
\label{sec:snr}
 \change{As described in Sec.~\ref{sec:sensing-frame},} each cell possesses an anchor node \change{performing} the \gls{snr} estimation \change{by processing the known pilot signals of length $L_p$ symbols transmitted by the satellites operating in K-band and above.} Therefore, the pilots transmitted in the \gls{dl} by the satellites are known by the anchor nodes and sensing takes place in the same frequency band $f_s$ used for \gls{dl} communication from satellite $s$.
 Let $m_i$, $z_i$, and $y_{i,s,c}$ be the $i$-th transmitted symbol, the complex sampled and zero-mean AWGN of unit variance, and the received symbol at cell $c$ from satellite $s$ in frame $k$, respectively. Then, the $i$-th received pilot symbol for cell $c$ is
 \begin{IEEEeqnarray}{rCl}
     y_{i,s,c} &=& m_i\sqrt{\gamma_{s,c}(k)\sigma^2} + z_i\sqrt{\sigma^2}.
  \end{IEEEeqnarray}
 The unbiased \gls{mle} for the \gls{snr} in a complex-valued channel, given that a known pilot signal of length $L_p$ is transmitted without upsampling at the receiver, is given as~\cite{Paluzzi2000}
\begin{equation}
    \hat{\gamma}_{s,c}(k)=\frac{\left(L_p-\frac{3}{2}\right)\left(\frac{1}{L_p}\sum_{i=1}^{L_p}\mathrm{Re}\left\{y_{i,s,c}^*\,m_i\right\}\right)^{\!\!2}}{\sum_{i=1}^{L_p}|y_{i,s,c}|^2\!-\!\frac{1}{L_p}\!\left(\sum_{i=1}^{L_p}\mathrm{Re}\left\{y_{i,s,c}^*\,m_i\right\}\!\right)^{\!\!2}}.
\end{equation}
Furthermore, the \gls{crb} for the \gls{snr} estimator in an \gls{awgn} channel is given as~\cite{Paluzzi2000}
\begin{equation}
    \text{var}\left(\hat{\gamma}_{s,c}(k)\right)\geq 3\gamma_{s,c}(k)/L_p.
    \label{eq:crb}
\end{equation}

The estimated \gls{snr} $\hat{\gamma}_{s,c}(k)$ can then be used to estimate the attenuation due to rain by inverting~\eqref{eq:snr} as
\begin{equation}
    \hat{A}_{s,c}(k)=\frac{P_s G_{c,s} G_{s}}{\set{L}_{s,c}(k) \hat{\gamma}_{s,c}(k)\,\ell\,\sigma^2}.
    \label{eq:att_est}
\end{equation}
Despite the \gls{mle} $\hat{\gamma}_{s,c}(k)$ being unbiased, the estimator $\hat{A}_{s,c}(k)$ is biased, as the second-order Taylor approximation gives
\begin{equation}
    \mathbb{E}\big[\hat{A}_{s,c}(k)\big]\approx A_{s,c}(k)\left(1+\text{var}\left(\hat{\gamma}_{s,c}(k)\right)\middle/\gamma_{s,c}(k)^2\right).
    \label{eq:2nd-order}
\end{equation}
%
%
Thus, we employ the bias correction method described in the Appendix, which uses the \gls{crb} for the \gls{snr} estimation to obtain the following estimator for the rain attenuation
\begin{equation} \label{eq:att-hathat}
    \begin{aligned}
        \hathat{A}_{s,c}(k) 
    =& \frac{P_s G_{c,s} G_{s}}{\set{L}_{s,c}(k)\ell\sigma^2\left(\hat{\gamma}_{s,c}(k)+3/L_p\right)}.
    \end{aligned}
\end{equation}

\subsection{\Gls{jmra} problem}
\label{sec:optimization}
To define the \gls{jmra} problem, let $x_{s,c}(k)\in\{0,\ldots, N_C\}$ be the number of \gls{ofdma} frames allocated for communication to the satellite-cell pair $(s,c)$ in frame $k$. Next, let $\mathbf{X}(k)\in\{0,\ldots,N_C N_B\}^{S(k)\times C}$ be the \gls{ra} matrix at frame $k$, whose $s$-th row is $\mathbf{x}_{s}(k)$ and its $(s,c)$-th element is $x_{s,c}(k)$. 

Then, to define a satellite-cell matching, we first define the adjacency indicator for cell $c$ and satellite $s$ at frame $k$ based on the definition of the Heaviside function:
$\alpha_{s,c}(k)=1$ when a satellite $s$ has allocated communication resources to cell $c$ in frame $k$, i.e., when $x_{s,c}(k)> 0$; otherwise, $\alpha_{s,c}(k)=0$. Next, we denote the set of cells served by satellite $s$ in frame $k$ as $\set{E}_s(k) = \{c:\alpha_{s,c}(k)=1\}$, where $|\set{E}_s(k)|\leq C_s$.

\begin{definition}\emph{Satellite-to-cell matching}.
A satellite-cell matching is a bipartite graph with vertex set $\set{S}(k)\cup\set{C}$ and edge set \begin{IEEEeqnarray*}{rCl}
    \set{E}(k)=\!\!\bigcup_{s\in\set{S}(k)} \!\!\set{E}_s(k)&=&\{(s,c):\alpha_{s,c}(k)=1,s\in\set{S}(k),c\in\set{C}\}
\end{IEEEeqnarray*}
\end{definition}

\change{To define the objective function, let $R_{s,c}(k)$ be the per-user throughput for cell $c \in \mathcal{C}$ served by satellite $s$ within frame $k$ for a given number of allocated \gls{ofdma} frames $x_{s,c}(k)$ and for a specific handover disconnection time 
\begin{equation} \label{eq:cost}
    H_{s,c}(k)=T_\text{HO}\left(1- \alpha_{s,c}\left(k-1\right)\right).
\end{equation}
Coherently with assumptions of Section~\ref{sec:comm-frame}, $H_{s,c}(k)$ is the time needed  for cell $c$ to complete the handover process in frame $k$ if satellite $s$ was not serving this cell in the previous frame. During this time, data communication is not possible~\cite{xie2021}.
Then, we calculate $R_{s,c}(k)$ as the product of the achievable rate $\rho_{s,c}(k)$ and the time available for data communication from satellite $s$ to the users in cell $c$, divided by the frame length $T_F$ and the number of active users $M_c$ as
\begin{equation}
    R_{s,c}(k) =\frac{\big(Tx_{s,c}(k)-H_{s,c}\left(k\right)\alpha_{s,c}(k)\big)\,\rho_{s,c}(k)}{ T_F M_c}.
    \label{eq:rate}
\end{equation}}

 
Our goal is to determine the values of $\mathbf{X}(k)$, the allocation of satellites' resources to the cells, to distribute them efficiently and fairly across the cells in the area while considering the cost of handovers. This is modeled as a load balancing problem, where the nominal data rate $\rho_{s,c}(k)$ for the individual cells is maximized when being served by the satellite with highest \gls{snr} $\gamma_{s,c}(k)$. However, due to the uneven geographical distribution of users in the cells, it might be convenient to connect to a satellite that is farther away but has more available resources or eliminates the need for a handover.

The following three constraints are defined for the resource allocation within one frame. First, one satellite can allocate up to $N_C$ \gls{ofdma} frames to each cell $c\in \mathcal{C}$. Second, each satellite cannot allocate more than $N_C N_B$ resources to the cells. Third, a cell $c \in \mathcal{C}$ can never be served by multiple satellites. This latter constraint defines the satellite-to-cell matching as a one-to-many matching problem. Therefore, up to one satellite can allocate between $1$ and $N_C$ \gls{ofdma} frames to a given cell $c$ and the rest of the satellites must allocate $0$ frames to the cell. With these constraints in place, we formulate the optimization problem as follows.
\begin{IEEEeqnarray}{CCll} \label{eq:opt_problem}
\mathcal{P}_1: \max_{\mathbf{X}(k)}&~&\IEEEeqnarraymulticol{2}{l}{\sum_{c\in\mathcal{C}}M_c\log\Bigg(1+\sum_{s\in\mathcal{S}}\!R_{s,c}(k)\!\Bigg),}\IEEEyesnumber*\IEEEeqnarraynumspace\label{eq:opt_problem_objective}\\[1em]
\text{subject to}&&
x_{s,c}(k)\in\{0,1,\dotsc, N_C\}, ~&\forall s,c\IEEEyessubnumber*\\[0.5em]
&&\displaystyle \sum_{c \in \mathcal{C}}x_{s,c}(k)\leq N_C N_B, &  \forall s \in \mathcal{S}(k),\IEEEyessubnumber*\\
&&\displaystyle\sum_{s \in \mathcal{S}(k)} \alpha_{s,c}\left(k\right)\in\{0,1\}, & \forall c \in \mathcal{C},~\IEEEyessubnumber*\IEEEeqnarraynumspace\label{eq:opt_problem_c3}
\end{IEEEeqnarray}
where the $\log(\cdot)$ function is applied to the objective function to promote a proportionally fair rate selection solution ~\cite{Mo2000:fair, Chen23}. \change{This proportional fairness is promoted by the logarithm in~\eqref{eq:opt_problem_objective}, which reduces the rate of growth of the objective as $R_{s,c}(k)$ increases. Consequently, increasing $x_{s,c}(k)$ for a cell $c$ with a low per-user throughput $R_{s,c}(k)$ results in a steeper increase in the objective when compared to increasing $x_{s,c}(k)$  for a cell with an already high $R_{s,c}(k)$.} Note that the value of $1$ is added inside the logarithm to achieve a feasible solution when $\sum_{s\in\mathcal{S}}\!R_{s,c}(k)=0$. \change{Furthermore, the logarithm is multiplied by the number of active users in the cells $c\in\mc{C}$ to account for the uneven distribution of users among the cells.}
 
 Finding a solution to $\mathcal{P}_1$ requires solving the multiple knapsack problem, which is NP-hard. The presence of $\alpha_{s,c}(k)$ inside $R_{s,c}(k)$ further complicates the problem since it is neither a convex nor a concave function of $x_{s,c}(k)$. Nevertheless, its presence is needed to avoid a negative value inside the logarithm due to $H_{s,c}(k)$. In addition, constraint~\eqref{eq:opt_problem_c3}, which ensures that each cell is connected to a single satellite, results in a one-to-many matching problem. Consequently, no efficient optimization algorithm exists to solve  problem $\mathcal{P}_1$.

\subsection{Proposed \gls{jmra} framework}\label{sec:jmra}
In the following, we reformulate problem $\mathcal{P}_1$ so its optimal solution can be closely approximated using a penalty method for convex optimization. This reformulation requires:
\begin{enumerate}
    \item Relaxing the objective and the first and second constraints to accommodate a real-valued optimization variable $\hat{\mathbf{X}}(k)\in\left[0,N_C\right]^{S(k)\times C}$ with its $(s,c)$-th element being $\hat{x}_{s,c}(k)$.
    \item Substituting $\alpha_{s,c}(k)$ with a continuous and convex approximation that eliminates the discontinuity of the Heaviside step function in $0$.
    \item Reformulating the optimization objective to include a \emph{penalty} (i.e., a set of weighting terms) that promotes the sparsity of the real-valued optimization variable $\hat{\mathbf{X}}(k)$ \change{by making~\eqref{eq:opt_problem_c3} an implicit constraint.}
\end{enumerate}

The continuous relaxation of $\set{P}_1$ to include $\hat{x}_{s,c}(k)$ as optimization variables is straightforward. Then, to approximate $\alpha_{s,c}(k)$, let us define the following inequality 
\begin{equation} 
    \frac{\hat{x}_{s,c}(k)}{\tau +\hat{x}_{s,c}(k)}\leq\alpha_{s,c}(k),
    \label{eq:ineq_alpha}
\end{equation}
where $0<\tau$. The left hand-side of~\eqref{eq:ineq_alpha} is a continuous \change{and convex} function of $\hat{x}_{s,c}(k)\geq 0$, which is a tight lower bound  for $\alpha_{s,c}(k)$ when  $\hat{x}_{s,c}(k)\gg \tau$. Next, to remove the optimization variable $\hat{x}_{s,c}$ from the denominator, we introduce an auxiliary variable $w_{s,c}\approx 1/(\tau+\hat{x}_{s,c}(k))$, which gives
\begin{equation}
w_{s,c}\,\hat{x}_{s,c}(k)\approx \alpha_{s,c}(k).
\label{eq:convex_approx}
\end{equation}
Finally, we substitute $\alpha_{s,c}(k)$ with $w_{s,c}\,\hat{x}_{s,c}(k)$ in~\eqref{eq:rate} to approximate $R_{s,c}(k)$ with the convex function
\begin{IEEEeqnarray}{c}
    \hat{R}_{s,c}(k) = \frac{\hat{x}_{s,c}(k)\,\rho_{s,c}(k)\left(T-H_{s,c}(k)\,w_{s,c}\right)}{ T_F M_c}\!\approx R_{s,c}(k).\IEEEeqnarraynumspace
\end{IEEEeqnarray}

As the final step, \change{we reformulate  $\set{P}_1$ using the \emph{augmented Lagrangian method}. This entails transforming the explicit constraint~\eqref{eq:opt_problem_c3} into an implicit constraint as the penalty for matching a cell to multiple satellites. Formulating the augmented Lagrangian requires the definition of a set of slack variables $\{\Sigma_c\}_{c\in\mc{C}}$ such that $\Sigma_c\in\{-1,0\}$ for all $c\in\mc{C}$ to transform the inequality constraint~\eqref{eq:opt_problem_c3} into an equality constraint using the approximation defined in~\eqref{eq:ineq_alpha} as}
\change{\begin{equation}
\sum_{s\in\mathcal{S}}w_{s,c}\,\hat{x}_{s,c}(k)+\Sigma_c=0, \quad\forall c\in\mc{C}.
\label{eq:c_convex_relax}
\end{equation}
Next, we define $\{\lambda_c\}_{c\in\mc{C}}$ to be the set of Lagrange multipliers and $\{p_c\}_{c\in\mc{C}}$ be the set of penalty terms for constraint~\eqref{eq:c_convex_relax}. Building on these, we redefine $\mc{P}_1$ into the convex problem}
\change{\begin{IEEEeqnarray}{CCll}
\!\mathcal{P}_2\!:\max_{\hat{\mathbf{X}}(k), \{\Sigma_c\}} & \IEEEeqnarraymulticol{2}{l}{\sum_{c\in\mathcal{C}}M_c\hspace{-.3mm}\log \!\Bigg(\hspace{-1mm}1+\hspace{-0.8mm}\sum_{s\in\mathcal{S}}\!\hat{R}_{s,c}(k)\!\hspace{-.8mm}\Bigg) \hspace{-.5mm}-\hspace{-.5mm} f_p\left(\hat{\mathbf{X}}(k),p_c\right)} 
 \IEEEeqnarraynumspace\label{eq:problem_cvx}\\[0.4em]
\text{subject to~}&
\displaystyle\!0\leq \hat{x}_{s,c}(k)\leq N_C, &~ \forall s \in \mathcal{S}(k),\, c\in\mathcal{C},\IEEEyessubnumber*\IEEEeqnarraynumspace\\[0.3em]
&\displaystyle \sum_{c \in \mathcal{C}}\hat{x}_{s,c}(k)\leq N_C N_B, & ~\forall s \in \mathcal{S}(k).\IEEEyessubnumber*\\
& -1\leq \Sigma_c\leq 0, & ~\forall c\in \mathcal{C},\IEEEyessubnumber*
\end{IEEEeqnarray}\vspace{-0.5em}}
\change{with a penalty function
\begin{equation} \label{eq:penalty}
\begin{aligned}
    f_p\left(\hat{\mathbf{X}}(k), p_c\right) = &~\frac{p_c}{2}\left(\sum_{s\in\mathcal{S}}w_{s,c}\,\hat{x}_{s,c}(k)+\Sigma_c\right)^2 \\ &+\lambda_c\left(\sum_{s\in\mathcal{S}}w_{s,c}\,\hat{x}_{s,c}(k)+\Sigma_c\right).
\end{aligned}    
\end{equation}
}
%
%
\change{\begin{proposition}
Problem $\mc{P}_2$ is equivalent to $\mc{P}_1$.
\end{proposition}
\begin{proof}
According to\cite{Boyd04, Griffin10}, problem $\mc{P}_2$ is equivalent to $\mc{P}_1$ if and only if the convex approximation in~\eqref{eq:convex_approx} is tight and the penalty function $f_p$ fulfills the following two conditions: $1$) it enforces feasibility of the problem in the asymptotic case where the penalty parameter $p_c\to \infty$; and $2$) it does not modify the objective function when all the implicit constraints, derived from the original problem $\mc{P}_1$, are fulfilled. In our maximization problem, this is achieved since
 \begin{equation}
     \lim_{p_c\to\infty} -f_p\left(\hat{\mathbf{X}}(k),p_c\right)=\begin{cases}
        0, & \text{if \eqref{eq:c_convex_relax} is fulfilled,}\\
         -\infty, & \text{otherwise}.
     \end{cases}
 \end{equation}
 Notably, if $w_{s,c}\hat{x}_{s,c}(k)=\alpha_{s,c}(k)$ and all constraints in~\eqref{eq:c_convex_relax} are fulfilled, $\sum_{s\in\mathcal{S}}w_{s,c}\,\hat{x}_{s,c}(k)\leq 1$ and the penalty~\eqref{eq:penalty} becomes $0$. Otherwise, if constraint~\eqref{eq:c_convex_relax} is not fulfilled for a cell $c$, $\sum_{s\in\mathcal{S}}w_{s,c}\,\hat{x}_{s,c}(k) >1$ and,  given $p_c\to\infty$, the penalty~\eqref{eq:penalty} becomes $-\infty$, completing the proof. 
\end{proof}
}
\change{We solve Problem~$\mc{P}_2$ using a combination of \emph{\gls{sca}} and the \emph{method of multipliers.} The former relies in iteratively approximating $\alpha_{s,c}(k)$ as defined in~\eqref{eq:convex_approx} by continuously updating $w_{s,c}$. On the other hand, the method of multipliers is an efficient iterative method to update the penalty terms $\{p_c\}$ and the Lagrange multipliers $\{\lambda_c\}$, to enforce the implicit constraint~\eqref{eq:c_convex_relax}. Specifically, these terms are first initialized to a pre-defined value. Then, at each iteration, the convex problem is solved, for example, using interior point methods. Afterwards, the weights are updated as $w_{s,c}=1/\left(\tau+\hat{x}_{s,c}(k)\right)$ and, if constraint~\eqref{eq:c_convex_relax} is violated for cell $c$ by more than a given tolerance $\Theta$, the penalty term $p_c$ and the Lagrange multiplier $\lambda_c$ are increased by a factor $\Delta$. This process is summarized in Algorithm~\ref{alg:1}.}

\change{Note that the values for the initialization of $w_{s,c}$, $p_c$, $\lambda_{c}$, and $\tau$ must be selected empirically. Furthermore, these are typically initialized to a considerably low value, so it is likely that some cells are initially matched to more than one satellite. Nevertheless, after several iterations of Algorithm~\ref{alg:1}, the values of $\hat{x}_{s,c}(k)$ converge to an optimal solution $\hat{\mathbf{X}}(k)^*$ for $\set{P}_2$ and, by increasing $p_c$ and $\lambda_{c}$, the number of cells matched to multiple satellites will be reduced.} However, there are no guarantees that all the cells are matched to up to one satellite after a finite number of iterations of Algorithm~\ref{alg:1}. In such cases, the solution is outside the feasible region of $\set{P}_1$ according to constraint~\eqref{eq:opt_problem_c3}. To avoid these problems, the real-valued optimal allocation $\hat{\mathbf{X}}(k)^*$ can be mapped to discrete values to obtain the final solution $\mathbf{X}(k)^*$ that fulfills the constraints set for $\mathcal{P}_1$ in~\eqref{eq:problem_cvx}. A common solution to map the real-valued allocation obtained from Algorithm~\ref{alg:1} is simply applying the rounding function to $\hat{x}_{s,c}(k)$. Afterwards, the simple procedure shown in Algorithm~\ref{alg:valid_allocation} is used to modify the allocation $\hat{\mathbf{X}}(k)$ so that it finds the closest point to $\hat{\mathbf{X}}(k)$ in the feasible region of the original problem $\mathcal{P}_1$.

\begin{algorithm}[t]
\caption{JMRA for global optimization of $\mc{P}_2$.}\label{alg:1}
\footnotesize
\setlength{\baselineskip}{1.2em}
\change{
\begin{algorithmic}[1]
\REQUIRE $\mathbf{X}(k-1)$, $\tau$, $\Theta$, $\Delta$, and $\rho_{s,c}(k)$ for all $s,c$
\STATE Initialize $\hat{x}_{s,c}(k)\!\leftarrow 0$, $x'_{s,c}\! \leftarrow 0$, $p_{c}\leftarrow 0$, $\lambda_c\leftarrow 0$, and $w_{s,c}$ for all $s,c$ 
\FOR {$n\in\{1,2,\dotsc,n_\text{iter}\}$}
\STATE Find $\hat{\mathbf{X}}(k)^*$ by solving ~\eqref{eq:problem_cvx}
\IF{$\sum_{s\in\mathcal{S}}w_{s,c}\,\hat{x}_{s,c}(k)-1\leq\Theta$ \AND $\left\lvert\hat{x}_{s,c}(k)-x'_{s,c}\right\rvert<\Theta,~ \forall c$ }
\STATE Convergence criteria fulfilled
\STATE \textbf{Break}
\ELSE
\STATE Update $\displaystyle w_{s,c}\leftarrow 1/\left(\tau+\hat{x}_{s,c}(k)\right)$ and $x'_{s,c}\leftarrow \hat{x}_{s,c}(k)$ for all $s,c$
\FOR {$c\in\mc{C}$} 
\IF{$\sum_{s\in\mathcal{S}}w_{s,c}\,\hat{x}_{s,c}(k)-1>\Theta$}
\STATE $p_c\leftarrow \Delta\cdot p_c$
\ENDIF
\STATE $\lambda_c\leftarrow \max\left\{0,\lambda_c+p_c\left(\sum_{s\in\mathcal{S}}w_{s,c}\,\hat{x}_{s,c}(k)-1\right)\right\}$
\ENDFOR
\ENDIF
\ENDFOR
\STATE $\mathbf{X}(k)^*\leftarrow \round\left(\hat{\mathbf{X}}(k)^*\right)$ 
\STATE Run Algorithm~\ref{alg:valid_allocation} on $\mathbf{X}(k)^*$
\RETURN $\mathbf{X}(k)^*$
\end{algorithmic}}
\end{algorithm}

\begin{algorithm}[t]
\footnotesize

\caption{Resource allocation adjustment.}
\label{alg:valid_allocation}
\setlength{\baselineskip}{1.2em}
\begin{algorithmic}[1]

\REQUIRE $\mathbf{X}(k)$, $\hat{\mathbf{X}}(k)$, and $\alpha_{s,c}(k)$ for all $s,c$
\FOR {$c\in\mathcal{C}$}
\IF {$\sum_{s\in\mathcal{S}(k)}\alpha_{s,c}(k)>1$}
\STATE $s^*\leftarrow \argmax_s R_{s,c}(k)$
\STATE $x_{s,c}(k)\leftarrow 0$ for all $s\neq s^*$
\ENDIF
\ENDFOR
\FOR {$s\in\mathcal{S}(k)$}
\WHILE{$\sum_{c\in\mathcal{C}}x_{s,c}(k)>N_TN_B$}
\STATE $c'\leftarrow \argmax_c x_{s,c}(k)-\hat{x}_{s,c}(k)$
\STATE $x_{s,c'}(k)\leftarrow x_{s,c'}(k)-1$
\ENDWHILE
\ENDFOR
\RETURN $\mathbf{X}(k)$
\end{algorithmic}
\end{algorithm}

\section{Performance analysis methodology} 
\label{sec:performance}
Our analyses focus on characterizing the performance of the proposed \gls{jmra} framework with realistic models for \gls{mle} \gls{snr} estimation, handover interruption time $T_\text{HO}$, and processing time required to solve the optimization problem. \change{In this section, we describe the benchmarks and the metrics used for the performance analysis performed in Sec.~\ref{sec:results}.}

The performance of our \gls{jmra} framework is compared to several benchmarks, divided into the following two categories: 

 \noindent\emph{1) Optimization framework:} The performance of the proposed \gls{jmra} framework described in Algorithms~\ref{alg:1} and~\ref{alg:valid_allocation} is compared to the performance of
 \begin{itemize}
    \item \emph{\Gls{jmra} without handover penalty (\gls{jmra} w/o HOP):}  The handover interruption time $T_\text{HO}$ is set to $0$ when running Algorithm~\ref{alg:1} and, therefore, it is neglected during optimization but included during evaluation.
    \item \emph{\Gls{dmrab} framework:} This follows the typical approach of solving the matching and resource allocation problems separately~\cite{Tang21, Gao21, Deyi2021}.  Therefore, the framework begins by finding a satellite-to-cell/user matching. Afterwards, resource allocation is performed at each satellite individually, with the pre-assigned cells. For this benchmark, we consider a satellite-to-cell one-to-many matching that solves the following optimization problem
    \begin{equation}
    \begin{aligned}
        \mathcal{P}_3: \max_{\{\alpha_{s,c}(k)\}}& \sum_{s\in\mathcal{S}(k)}\sum_{c\in\mathcal{C}_s(k)}\rho_{s,c}(k)\,\alpha_{s,c}(k), \\ 
        \text{subject to } & \eqref{eq:opt_problem_c3}.
    \end{aligned}        
    \end{equation}
     Afterwards, an optimal resource allocation is performed at each satellite to solve the problem
     \begin{IEEEeqnarray}{CCll} \label{eq:opt_problem:bench:ra}
        \mathcal{P}_4: \max_{\mathbf{x}_s(k)}&~&\IEEEeqnarraymulticol{2}{l}{\sum_{c\in\mathcal{C}_s(k)}M_c\log\Big(\!1+R_{s,c}(k)
        \!\Big)},\IEEEyesnumber*\IEEEeqnarraynumspace\label{eq:disj_opt_problem_objective}\\
        \text{subject to}&&
        0\leq x_{s,c}(k)\leq N_C, ~&\forall c\in\mathcal{E}_s(k)\IEEEyessubnumber*\\[0.5em]
        &&\displaystyle \sum_{c \in \mathcal{C}}x_{s,c}(k)\leq N_C N_B,\IEEEyessubnumber*   
    \end{IEEEeqnarray} 
    The \gls{dmrab} framework is described in Algorithm~\ref{alg:dist}.  
\end{itemize}

\begin{algorithm}[tb]
\caption{\Gls{dmrab}.}\label{alg:dist}
\footnotesize
\setlength{\baselineskip}{1.2em}
\begin{algorithmic}[1]
\REQUIRE $\rho_{s,c}(k)$ for all $s,c$
\STATE Initialize $\hat{x}_{s,c}(k)\leftarrow 0$ and $\alpha_{s,c}(k)\leftarrow 0$ for all $s,c$ 
\STATE Initialize $\set{E}_s(k)=\emptyset$ for all $s$
\FOR {$c\in\set{C}$}
\STATE $s^*\leftarrow \argmax_s \rho_{s,c}(k)$
\STATE $\alpha_{s^*,c} (k)\leftarrow 1$
\STATE $\set{E}_s(k)\leftarrow \set{E}_s(k)\cup s^*$ 
\ENDFOR
\FOR {$s\in\set{S}(k)$}
\STATE Find $\hat{\mathbf{x}}_{s,c}(k)^*$ for $\set{E}_s(k)$ by solving~\eqref{eq:disj_opt_problem_objective}
\STATE $\mathbf{x}_k^*\leftarrow \round\left(\hat{\mathbf{x}}_k^*\right)$ 
\WHILE{$\sum_{c\in\mathcal{C}}x_{s,c}(k)>N_C N_B$}
\STATE $c'\leftarrow \argmax_c x_{s,c}(k)-\hat{x}_{s,c}(k)$
\STATE $x_{s,c'}(k)\leftarrow x_{s,c'}(k)-1$
\ENDWHILE
\ENDFOR
\RETURN $\mathbf{X}(k)^*$
\end{algorithmic}
\end{algorithm}

\noindent\emph{2) Sensing framework:} The performance of the \gls{jmra} and \gls{dmrab} frameworks with the realistic sensing framework described in Sec.~\ref{sec:snr} is compared with
\begin{itemize}
    \item \emph{Perfect \gls{csi}:} Idealized case where the satellites obtain a perfect \gls{csi}, including the \gls{snr} $\gamma_{s,c}(k)$ and rain attenuation $A_{s,c}(k)$ from an external source (i.e., without sensing). This corresponds to an upper bound in performance that serves as a benchmark for the \gls{isac} framework.
    \item \emph{No sensing:} Na\"ive case where no sensing is performed and the \gls{snr} for all satellite-cell pairs $(s,c)$ is calculated from~\eqref{eq:snr} assuming  $A_{s,c}(k)=1$.
\end{itemize}





\subsection{\Glspl{kpi}}
As \glspl{kpi} for communication, we consider the average throughput per user and the overall fairness of the \gls{ntn}.
The former is given as $\overline{R} = \frac{1}{K}\sum_{k=0}^K\sum_{c\in\mc{C}}\sum_{s\in\mc{S}(k)} R_{s,c}(k)$
with $R_{s,c}(k)$ defined in~\eqref{eq:rate}, and $K$ the number of system frames under test.
For the latter, we use Jain's fairness index for individual frames, defined as
\begin{equation}
\mathcal{J}_k = \frac{\Big(\sum_{c\in\mathcal{C}}M_c\sum_{s\in\set{S}(k)}R_{s,c}(k)\Big)^2}{\Big(\sum_{c\in\mathcal{C}}M_c\Big)\sum_{c\in\mathcal{C}}M_c\Big(\sum_{s\in\set{S}(k)}R_{s,c}(k)\Big)^2}.
\end{equation}

As \gls{kpi} for sensing, we consider the \gls{nmse} for the \gls{snr} $\hat{\gamma}_{s,c}(k)$ and rain attenuation $\hathat{A}_{s,c}(k)$ estimators. The \gls{nmse} for $\hat{\gamma}_{s,c}(k)$ is
\begin{equation}
    \text{NMSE}_\gamma = \frac{\sum_{k=1}^K\sum_{c\in\set{C}}\sum_{s\in\set{S}(k)}\left(\gamma_{s,c}(k) - \hat{\gamma}_{s,c}(k) \right)^2}{\sum_{k=1}^K\sum_{c\in\set{C}}\sum_{s\in\set{S}(k)}\gamma_{s,c}(k)^2},
\end{equation}
and an analogous formulation is used for $\hathat{A}_{s,c}(k)$.

\subsection{Complexity analysis}
\textbf{Algorithm~\ref{alg:1}}: The complexity of solving $\set{P}_2$ once, for fixed values of $w_{s,c}$ (line 3) using an interior point method is $\set{O}\left(S(k)^3C^3\right)$. Since this process is repeated $n_\text{iter}$ times (lines 2 to 5), the solution of $\set{P}_2$ has a complexity $\set{O}\left(S(k)^3C^3n_\text{iter}\right)$.

\textbf{Algorithm~\ref{alg:dist}}: The complexity of solving the disjoint optimization problem is determined by the complexity of using an interior point method to solve $\set{P}_4$ at each satellite. In the worst case, one satellite will cover $C_s$ cells and, hence, solving the \gls{ra} for one satellite has a complexity $\set{O}\left(C_s^3\right)$. Since this must be preformed for each satellite, the complexity is $\set{O}\left(S(k)C_s^3\right)$.

\section{Results}
\label{sec:results}
We consider a rectangular area covering central Europe (the same shown in Fig.~\ref{fig:maps}), between latitudes $40^\circ$ and $55^\circ$ North  and longitudes $5^\circ$ and $30^\circ$ East. The cells are evenly spaced, each covering $0^\circ 15^\prime$ in both latitude and longitude \change{and the population of each cell $M^\text{max}_c$ is obtained from~\cite{CIESIN}}. Therefore, there are a total of $6161$ cells in the area, out of which $766$ have zero users because they are over the sea or entirely unpopulated areas. \change{Table~\ref{tab:sim_params} presents the simulation parameters used for performance evaluation. The parameters for the  satellite constellation correspond to those of the Starlink orbital shells at $550$ and $570$ km of altitude. The duration of the \gls{ofdma} frames and most of the communication parameters correspond to those in 3GPP 5G \gls{ntn} technical reports~\cite{3GPPTR38.821, 3GPPTR38.811}. The rain parameters were adapted to the discrete-time model presented in Section~\ref{sec:rain} from~\cite{Cow95}. Finally, to capture the dynamics of the environment, the length of the system frames is set to $T_F=\{10,30\}$\,s, which is shorter than or equal to typical values considered in the \gls{leo} satellite literature~\cite{Zhao24} and more than $10\times$ shorter than the period between required handovers in quasi-Earth fixed cell scenarios~\cite{Wigard23}.}

The results presented in this section were obtained using a simulator coded in Python to replicate the orbital dynamics of the satellites, the geographical distribution of the population, and the dynamics of the rain model~\cite{Cow95}. Each simulation comprises at least $100$ frames and the optimization problems were solved using the CVXPY package~\cite{diamond2016cvxpy} using MOSEK ApS as solver. 

\begin{table}[t]
\renewcommand{\arraystretch}{1.1}
\centering
\caption{Simulation parameters}
\begin{tabular}{@{}llcc@{}}
    \toprule \textbf{Parameter} & \textbf{Symbol} & \multicolumn{2}{c}{\textbf{Value}} \\\midrule 
   \textbf{Satellite constellation} && S-band & Ka-band\\\midrule
    Center frequency [GHz] & $f_i$ & $2.185$ & $19.95$\\
    Bandwidth [MHz]  & $B_i$ & $30$ & $500$\\
     Satellite antenna gain [dBi] & $G_{s,c}$ & $24$ & $30.5$\\
    Total number of satellites & $S_i$ & $720$ & $1584$\\
     Number of orbital planes & $P_i$  & $36$& $72$\\
     Altitude of deployment [km] & $h_i$  & $570$& $550$\\
     Inclination [deg] & $\delta$  & $70$& $53$\\
     Transmission power [W] &$P_\text{tx}$ & $75$ & $75$\\ 
     Pointing loss [dB]  & $\ell_\text{dB}$ & $0.3$ & $0.3$\\
     Minimum elevation angle [deg] & $\eta$ & $25$ & $25$\\ 
     \change{Number of beams per satellite} & \change{$N_B$} & \change{$19$} & \change{$19$}\\
    \midrule
    \textbf{Ground segment}\\
    \midrule
    Number of cells in the area & $C$ & \multicolumn{2}{c}{$6161$}\\
    User/anchor node antenna gain [dBi] & $G_{c,s}$& \multicolumn{2}{c}{$0$}\\
    Noise spectral density [dBm/Hz] & $N_{0,\text{dB}}$ & \multicolumn{2}{c}{$-176.31$}\\
    \change{Population per cell $c$} & \change{$M^\text{max}_c$} & \multicolumn{2}{c}{\change{From~\cite{CIESIN}}}\\
    Ratio of active users & $\mu_c$ & \multicolumn{2}{c}{$0.001$}\\
\midrule 
\textbf{Rain parameters}\\\midrule
     Rain cell intensity (PPP) [rain cells/km$^2$] & $\lambda_\text{rain}$ &\multicolumn{2}{c}{$8.4\times10^{-4}$}\\
     Rain height [km] & $h_r$ & \multicolumn{2}{c}{$6$}\\
     Mean rain intensity & $\overline{\varrho}$ & \multicolumn{2}{c}{$8.77$}\\
     Mean rain cell radius [km] & $d_\text{rain}$ & \multicolumn{2}{c}{$22.6$}\\
     Mean duration of a rain episode [h] & $\varepsilon$ & \multicolumn{2}{c}{$1.886$}\\
     Mean period between rain episodes [h] & $\beta$ & \multicolumn{2}{c}{$5.376$}\\
     Probability that a rain cell is active& $\pi_\text{on}$ & \multicolumn{2}{c}{$0.26$}\\
     \midrule
\textbf{Frame structure}\\
\midrule
    Frame duration [s] & $T_F$ &\multicolumn{2}{c}{$\{10,30\}$}\\
     \gls{ofdma} frame duration [ms] & $T$ & \multicolumn{2}{c}{$10$}\\
     Disconnection time due to handover [ms]& $T_\text{HO}$ & \multicolumn{2}{c}{$\{50, 100\}$}\\
     Pilot length [symbols] & $L_p$ & \multicolumn{2}{c}{$\{2^2,2^4,\dotsc,2^{16}\}$}\\
     \bottomrule
\end{tabular}
\label{tab:sim_params}
\vspace{-10pt}
\end{table}


\change{\paragraph{Convergence, real-time implementation considerations, and parameter tuning}
First, we evaluate the convergence properties of Algorithm~\ref{alg:1} and the impact of the number of iterations on performance and execution time. Recall that, Algorithm~\ref{alg:1} reaches convergence, as indicated in line 4, when the violation of constrain~\eqref{eq:c_convex_relax} and the change in resource allocation between two iterations for all cells $c$ are below a pre-defined tolerance $\Theta$, which was set to $0.01$. To evaluate convergence, Fig.~\ref{fig:convergence_CDF} shows the empirical \gls{cdf} of the number of iterations needed to reach convergence for different values of parameter $\Delta$, which defines the increase in the penalty $p_c$ when constraint~\eqref{eq:c_convex_relax} is violated. As it can be seen, convergence was reached for all the settings of $\Delta$ and a higher value of $\Delta=10$ decreases the number of iterations needed to reach convergence. However, increasing $\Delta$ must be done with caution, as it might result in infeasible solutions and/or to numerical issues with the optimization software. Next, Fig.~\ref{fig:exec_time_double_matchings} compares the execution time and the performance achieved by Algorithm~\ref{alg:1} until convergence with those for a reduced number of maximum iterations $n_\text{iter}\leq 5$. These results were obtained by executing Algorithm~\ref{alg:1} in a MacBook Pro 2021 with an M1 Pro CPU and 16 GB of RAM. As expected, reducing the number of iterations of Algorithm~\ref{alg:1} results in a linear decrease of its execution time but also increases the amount of cells matched to multiple satellites and, thus, violating constraint~\eqref{eq:c_convex_relax}. Moreover, both the execution time and the cells violating the implicit constraint~\eqref{eq:c_convex_relax} are lower for $\Delta=10$ than for $\Delta=5$. These are expected results, since the number of iterations to reach convergence is lower for $\Delta=10$ than for $\Delta=5$, and a higher value of $\Delta$ imposes a higher penalty to cells violating~\eqref{eq:c_convex_relax}. In a real-time implementation, however, since the frame length $T_F$ determines the period between two consecutive executions of Algorithm~\ref{alg:1}, $n_\text{iter}$ must be adequately selected to achieve an execution time that satisfies condition~\eqref{eq:ra-time}, stated in Section~\ref{sec:ra-frame}. Thus, to validate a real-time implementation of Algorithm~\ref{alg:1}, we considered $\Delta=10$ and obtained that the mean execution time of Algorithm~\ref{alg:1} until convergence is
$T_\mathrm{RA} = 18.184\pm 0.319$\,s with $95\%$ confidence. 
Then, by considering $T_F = 30$\,s, and the longest sensing and overhead times $T_S = 600$\,ms and $T_\mathrm{HO} = 100$\,ms of the paper, we calculated the probability of the execution time $T_\mathrm{RA,2}$ exceeding the worst case of $T_F - T_S - T_\mathrm{HO} = 29.3$\,s to be $1.313 \cdot 10^{-6}$. For this, we fitted a Gamma distribution to the execution times, obtaining a goodness of fit with p-value $0.699$, and neglected $T_\mathrm{RA,1}$ and $T_\mathrm{RA,3}$, which are much shorter than $T_\mathrm{RA,2}$. This allowed us to conclude that Algorithm~\ref{alg:1} can achieve convergence in a real-time implementation in a device with similar or greater processing capabilities than those of the one mentioned above for $T_F=30$\,s. Conversely, achieving a real-time implementation for $T_F=10$\,s with our software implementation and processing hardware is only possible for $n_\text{iter}\leq2$. Nevertheless, convergence can be achieved in real-time through hardware and/or software improvements. Therefore, we set $\Delta=10$ and execute Algorithm~\ref{alg:1} until convergence for the rest of the results presented in this section.}

\change{Afterwards, we identify an adequate parameter setting for the initialization of $w_{s,c}$. Our experiments revealed that the specific setting has a minor impact on performance as long as these are initialized to the same value $w_{s,c}\ll 1$ for all $(s,c)$. Thus, to obtain the results presented in this paper, we initialized $w_{s,c}\leftarrow C/(N_C \,N_B\,S(k))$, which is a non-informative prior that assumes that equal resources are given to all the cells. In the best case, the latter initialization leads to a $2.45$\% decrease in the number of handovers compared to initializing $w_{s,c}$ to $0$.}

\begin{figure}[t]
    \centering
\includegraphics{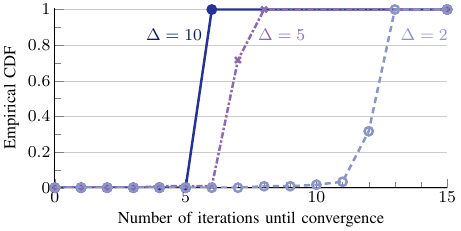}
    \caption{\change{Empirical CDF of the number of iterations needed to achieve convergence for different values of parameter $\Delta=\{2,5,10\}$, which determines the increase of $p_c$ at each iteration.}}
    \label{fig:convergence_CDF}
\end{figure}

\begin{figure}[t]
    \centering
\subfloat[]{\includegraphics{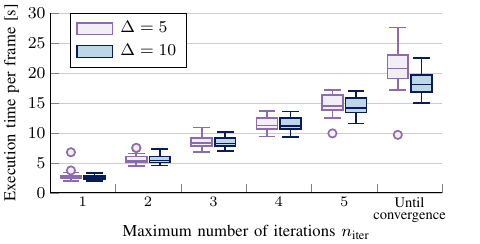}}\\
\subfloat[]{\includegraphics{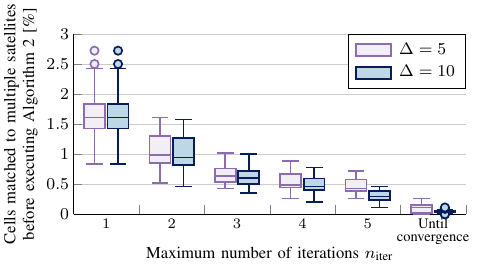}}
 \caption{\change{Box plots of the  (a) execution time of Algorithm~\ref{alg:1} and (b) percentage of cells violating constraint~\eqref{eq:opt_problem_c3} for different values of $n_\text{iter}$ and $\Delta$.}}
  \label{fig:exec_time_double_matchings}
  \vspace{-1em}
\end{figure}

\paragraph{Impact of handovers in short frames}
After finding adequate parameters for Algorithm~\ref{alg:1}, evaluate the impact of the handover disconnection time $T_\text{HO}$ on performance \change{by considering  $T_\text{HO}=\{50,100\}$\,ms, along with a perfect \gls{csi}. Furthermore, we consider two typical frame lengths $T_F=\{10,30\}$\,s, which are sufficiently short to guarantee that the change in path loss, from the beginning to the end of the frame, for any satellite-cell pair is below $0.574$ dB for $T_F=10$~s and $1.711$~dB for $T_F=30$~s.} Our results with the \gls{jmra} framework and two benchmarks are presented in Table~\ref{tab:results_short_frames}. 

Table~\ref{tab:results_short_frames} shows that our proposed \gls{jmra} framework greatly outperforms \gls{dmrab}, increasing the throughput by up to $59$\% and the fairness index by more than $600$\%. Furthermore, \gls{dmrab} cannot find a feasible solution for $\set{P}_4$ with $T_F=10$\,s and $T_\text{HO}=100$\,ms, as the resources in S-band satellites matched to a large number of cells are insufficient to compensate for $T_\text{HO}$. Additionally, \gls{jmra} improves both the throughput and fairness compared to \gls{jmra} w/o HOP, with the highest improvement observed with $T_\text{HO}=100$\,ms and $T_F=10$\,s. 

The superior performance of \gls{jmra} is due to the careful satellite-to-cell matching to account for $T_\text{HO}$, which is reflected in the average number of handovers per second. Namely, \gls{jmra} adapts the number of handovers depending on $T_\text{HO}$, whereas the two benchmarks reach the exact same solution for a given $T_F$ and, consequently, the same average number of handovers, regardless of the value of $T_\text{HO}$. Furthermore, while the increase in performance of \gls{jmra} vs. \gls{jmra} w/o HOP is modest, it comes with no additional overhead and with no added complexity of Algorithm~\ref{alg:1}. Consequently, \gls{jmra} w/o HOP is not considered further.

\begin{table*}[htb]
\centering
\caption{Impact of $T_\text{HO}$ on the performance of the proposed \gls{jmra} framework and two benchmarks for different $T_F$ considering perfect \gls{csi}.}
\begin{tabular}{@{}lcrrrrrrrr@{}}
\toprule
& & \multicolumn{2}{c}{\textbf{Per-user throughput [kbps]}} &~& \multicolumn{2}{c}{\textbf{Fairness index}} &~& \multicolumn{2}{c}{\textbf{Handovers per second}}\\\cmidrule{3-4} \cmidrule{6-7} \cmidrule{9-10}
\textbf{Optimization framework}& $T_\text{HO}$\,[ms]&$T_F=10$\,s &$T_F=30$\,s &&$T_F=10$\,s  &$T_F=30$\,s&&$T_F=10$\,s &$T_F=30$\,s\\\midrule
~\gls{dmrab} & & $101.654$ & $99.486$ & & $0.126$ & $0.117$ & & $101.953$ & $82.946$\\
~\change{\gls{jmra} w/o HOP} & $50$& \change{$164.445$} &\change{$166.594$}&& \change{$0.784$}&\change{$0.790$} && \change{$148.554$}& \change{$118.707$}\\
~\change{\textbf{\gls{jmra}}}&  & \change{$163.996$} & \change{$166.586$} &&\change{$0.790$} & \change{$0.791$} & & \change{$125.738$} & \change{$114.267$}\\\midrule
~\gls{dmrab} & & -- & $97.855$ & &  -- & $0.115$ & & -- & $82.946$\\
\change{~\gls{jmra} w/o HOP} &  $100$& \change{$162.058$}& \change{$164.635$} &&\change{$0.777$} & \change{$0.789$} &&\change{$148.554$}& \change{$118.707$}\\
~\change{\textbf{\gls{jmra}}}&  & \change{$160.953$} & \change{$164.472$} && \change{$0.792$} & \change{$0.794$} & & \change{$111.161$} & \change{$111.055$}\\
\bottomrule
\end{tabular}
\label{tab:results_short_frames}
\vspace{-6pt}
\end{table*}

\paragraph{Impact of pilot length on performance}
Next, we evaluate the impact of the pilot length $L_p$ on the performance of sensing and communication for $T_\text{HO}=50$\,ms. 

To evaluate the sensing performance, Fig.~\ref{fig:nmse} shows the \gls{nmse} for the \gls{snr} estimation $\hat{\gamma}_{s,c}(k)$ and the estimated attenuation due to rain $\hathat{A}_{s,c}(k)$ for each satellite-cell pair for different pilot lengths $L_p$. The results for cells with no rain ($\varrho_c=0$) and with rain ($\varrho_c>0$) are shown with two separate curves. Note that the results for $\hat{A}_{s,c}(k)$ are not included as the bias for this estimator is so large that the \gls{nmse} exceeds $10^6$ for all the considered values of $L_p$.
As it can be seen, the \gls{nmse} of $\hat{\gamma}_{s,c}(k)$ decreases rapidly and follows a similar trend for cells with and without rain. On the other hand, there is a major difference between the \gls{nmse} of $\hathat{A}_{s,c}(k)$ with and without rain. Namely, in the absence of rain, $A_{s,c}(k)=1$ and the \gls{nmse} for $\hathat{A}_{s,c}(k)$ is always below $10^{-1}$. On the other hand, the \gls{nmse} of $\hathat{A}_{s,c}(k)$ for cells with rain is close to one for most values of $L_p$ and only a noticeable decrease is obtained for $L_p\geq 2^{14}$. To understand the reason for this behavior, it is beneficial to observe Fig.~\ref{fig:nmse_snr}, which shows the evolution of the \gls{nmse} for the \gls{mle} $\hat{\gamma}_{s,c}(k)$ at different \gls{snr} levels along with the \gls{crb}. While the \gls{nmse} of the \gls{crb} decreases as the \gls{snr} decreases, the \gls{mle} exhibits the opposite behavior. Thus, the error with the \gls{mle} increases rapidly as the \gls{snr} decreases. Naturally, as cells with rain experience an higher attenuation, the \gls{nmse} of $\hat{\gamma}_{s,c}(k)$ is higher for cells with rain. While this difference is small, it has a major impact on the estimator $\hathat{A}_{s,c}(k)$, as observed in Fig.~\ref{fig:nmse}. 

\begin{figure}[t]
\centering
\includegraphics{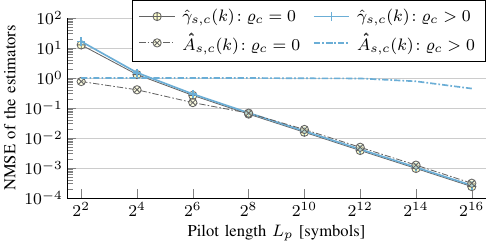}
\caption{\Gls{nmse} for the \gls{mle} of the \gls{snr} $\hat{\gamma}_{s,c}(k)$ and the attenuation due to rain $\hathat{A}_{s,c}(k)$ as a function of the pilot length $L_p$. Results are shown for cells with rain $\varrho_c(k)>0$ and cells with no rain $\varrho_c(k)=0$.}
\label{fig:nmse}
\end{figure}

\begin{figure}[t]
\vspace{-10pt}
\centering
\includegraphics{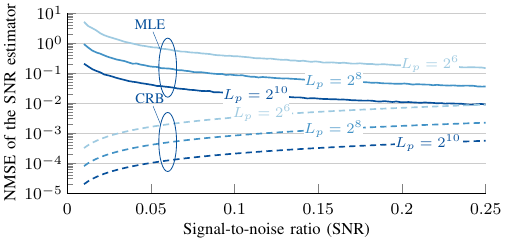}
\caption{\Gls{nmse} for the \gls{snr} estimation with the \gls{mle} $\hat{\gamma}_{s,c}(k)$ and the \gls{crb} for pilot lengths $L_p= \{2^6, 2^8, 2^{10}\}$.}
\label{fig:nmse_snr}
\end{figure}

To evaluate the communication performance, we show the average throughput and fairness index for the \gls{dmrab} and the proposed \gls{jmra} framework for $L_p= [2^2, 2^{16}]$ and the no sensing case in Fig.~\ref{fig:perf_vs_lp}. When sensing is performed, the throughput achieved by both \gls{dmrab} and \gls{jmra}, shown in Fig.~\ref{fig:th_lp} for $T_F=\{10,30\}$\,s, increases with the pilot length until $L_p=2^8$. The reason for this behavior is that the sensing frame is $T_S = 20$\,ms for all $L_p\leq 2^8$. Then, $T_S$ increases for all $L_p>2^8$ until reaching $T_S = 600$\,ms for $L_p=2^{16}$, which greatly decreases the throughput for \gls{dmrab}. 
Note that the performance of \gls{dmrab} exhibits major variations both in terms of the average throughput per user shown in Fig.~\ref{fig:th_lp} and in fairness shown by the box plot in Fig.~\ref{fig:fair_lp}. Namely, its fairness is greatly variable both across different frames with the same pilot length $L_p$ and across different values of the pilot length. This behavior arises from its inability to consider the impact of $T_\text{HO}$ and to the large amount of cells with rain that prefer to connect to S-band satellites.

On the other hand, it is clear that the performance of \gls{jmra} increases with the pilot length $L_p$ both in terms of throughput and fairness index. Specifically, the \gls{jmra} with $L_p=2^{16}$ and $T_F=30$\,s achieves a \change{$17.8$\%} increase in throughput when compared to the case without sensing, which accurately captures the \gls{snr} for cells without rain but not of those with rain. The only case where an increase of $L_p$ results in a decrease in performance is for $T_F=10$\,s and $L_p=2^{16}$. Even though is not included in Fig.~\ref{fig:perf_vs_lp}, a similar decrease occurs for $T_F=30$\,s and $L_p > 2^{16}$. The reason for this behavior is the increased length of the sensing phase, which would exceed $T_S=600$\,ms, accounting for more than $6$\% of the total frame duration. Fig.~\ref{fig:fair_lp} also shows that the fairness index for \gls{jmra} is greatly consistent, outperforming the case without sensing for $L_p\geq 2^4$ and presenting only slight variations both across different frames with the same pilot length $L_p$ and across different values of the pilot length, showcasing the vast superiority of \gls{jmra}, \change{which achieves up to a $700$\% increase in fairness over \gls{dmrab}.} 

\begin{figure}[t]
    \centering
    \subfloat[]{\includegraphics[]{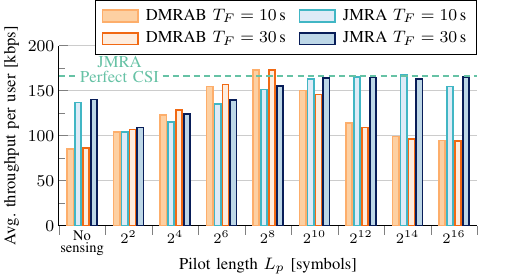}\vspace{-0.3em}\label{fig:th_lp}}\\[0.5em]
    \subfloat[]{\includegraphics[]{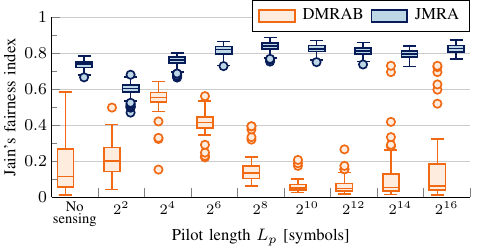}\vspace{-0.3em}\label{fig:fair_lp}}
    \caption{(a) Average per-user throughput for $T_F=\{10, 30\}$\,s and (b) Jain's fairness index per frame with $T_F=30$\,s as a function of the pilot length $L_p$.}
    \label{fig:perf_vs_lp}
\end{figure}

\paragraph{Impact of the frame length}
As discussed in Sec.~\ref{sec:protocol}, the duration of a system frame is constrained by the variation of the \gls{gsl}. In Fig.~\ref{fig:long_frames}, we illustrate this aspect by showing the impact of the frame length on the communication performance. It can be seen that the performance of both \gls{jmra} and \gls{dmrab} remains stable for $T_F \le 110$\,s, but then the throughput (Fig.~\ref{fig:long_frames_th}) and fairness (Fig.~\ref{fig:long_frames_fair}) plummet. This is because the movement of the satellites generates a significant change in the \gls{gsl} path within the frame for $T_F > 110$\,s, which greatly restricts the space of feasible solutions in the optimization. 

\begin{figure}[bth]
    \centering
    \subfloat[]{\includegraphics{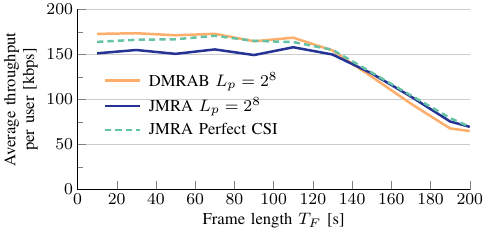}\vspace{-6pt}\label{fig:long_frames_th}}\\
    \subfloat[]{\includegraphics{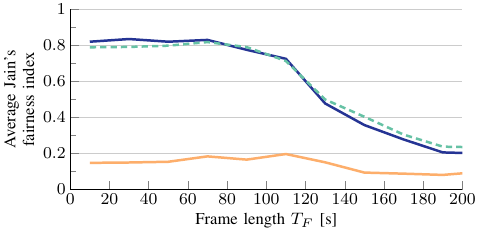}\vspace{-6pt}\label{fig:long_frames_fair}}
    \caption{Average per-user throughput for the \gls{dmrab} with $L_p=2^8$ and with our proposed \gls{jmra} with $L_p=2^8$ and with perfect CSI for $T_F\in\left[10,200\right]$.}
    \label{fig:long_frames}
\end{figure}

\section{Conclusions}
\label{sec:conclusions}
In this paper, we proposed an \gls{isac} framework and a candidate frame structure to achieve an efficient and fair resource allocation in \gls{leo} satellite constellations across wide areas. By jointly solving the satellite-to-cell matching and resource allocation problems, the proposed framework achieves \change{up to a $700$\% increase in fairness} index w.r.t. the selected benchmark, where these problems are solved independently. Furthermore, by carefully selecting the length of the pilot signal for sensing, the proposed \gls{isac} framework can achieve a per-user throughput that is less than $1$\% lower than the one achieved by the upper bound, which assumes perfect \gls{csi} with no signaling overhead. Moreover, we observed that setting a frame length of $10$ to $30$ seconds is adequate to adapt to the dynamic network topology and to the changes in the rain patterns without incurring in excessive signaling overhead due to handovers and \gls{csi} acquisition. Finally, we confirmed that our solution can be implemented in real time at a central server, even with the limited processing power of a typical PC. Future work includes refinements to distributed resource allocation algorithms to achieve a higher fairness, as well as mechanisms for satellite-based digital twinning of the atmospheric conditions. 

\appendix
\label{sec:appendix}
\change{As it can be seen} in~\eqref{eq:att_est}, the estimator $\hat{A}_{s,c}(k)$ is a random variable whose distribution is the inverse distribution of random variable $\hat{\gamma}_{s,c}(k)$ multiplied by the \gls{snr} with no rain attenuation, which is a known constant. Since the distribution of $\hat{\gamma}_{s,c}(k)$ is not known, it is not trivial to characterize  $\mathbb{E}\big[\hat{A}_{s,c}(k)\big]$. Nevertheless, from the second-order Taylor approximation~\eqref{eq:2nd-order}, we know that the estimator $\hat{A}_{s,c}(k)$ is biased.
On the other hand, the estimator
\begin{equation}
    \hat{A}^*_{s,c}(k)\triangleq \hat{A}_{s,c}(k) / \left(\mathbb{E}\big[\hat{A}_{s,c}(k)\big]/A_{s,c}(k) \right)
    \label{eq:opt_est}
\end{equation}
would clearly be unbiased, but it is infeasible, since the unknown parameter that must be estimated $A_{s,c}(k)$ is present in the denominator. 
Nevertheless, it is feasible to use the approximation of $\mathbb{E}\big[\hat{A}_{s,c}(k)\big]$ from~\eqref{eq:2nd-order}
to we define the refined estimator $\hathat{A}_{s,c}(k)$. This is done by by 1) substituting $\text{var}(\gamma_{s,c}(k))$ in~\eqref{eq:2nd-order} with the \gls{crb} defined in~\eqref{eq:crb}; 2) substituting $\gamma_{s,c}(k)$ with $\hat{\gamma}_{s,c}(k)$ as an approximation; and 3) substituting $A_{s,c}(k)$ with $\hat{A}_{s,c}(k)$ in~\eqref{eq:2nd-order} and~\eqref{eq:opt_est}, obtaining
\begin{IEEEeqnarray}{rCl}
    \hathat{A}_{s,c}(k)&\triangleq& \hat{A}_{s,c}(k) / \left(1+3\middle/L_p\hat{\gamma}_{s,c}(k)\right).
\end{IEEEeqnarray}
By substituting $\hat{A}_{s,c}(k)$ with its definition given in~\eqref{eq:att_est}, we obtain the estimator given in~\eqref{eq:att-hathat}. \hfill\qedsymbol

\bibliographystyle{IEEEtran}
\bibliography{bib.bib}

\end{document}